\documentclass[reqno]{amsart}
\usepackage[foot]{amsaddr}
\usepackage[UKenglish]{babel}
\usepackage{microtype}
\usepackage[small]{complexity}
\usepackage{relsize}
\usepackage{enumerate}
\newif\ifsubmission
\submissiontrue
\newif\ifams
\amstrue
\let\olddesc\description
\let\oldenddesc\enddescription
\providecommand{\IEEEsetlabelwidth}[1]{}
\renewenvironment{description}[1][]{\olddesc}{\oldenddesc}
\renewcommand{\paragraph}{\subsubsection*}
\makeatother
\usepackage{amssymb,mathtools,amsmath,amsthm,mathrsfs,upgreek}
\usepackage{threeparttable,booktabs}
\usepackage{listings}
\usepackage[numbers,compress]{natbib}
\makeatletter
\def\NAT@spacechar{~}%
\makeatother
\usepackage{dsfont}
\newcommand{\one}{\mathds{1}}
\usepackage{thmtools}
\renewcommand{\thmcontinues}[1]{%
    \hyperref[#1]{continued}%
}%
\declaretheorem[numberwithin=section]{theorem}

\declaretheorem[sibling=theorem]{property}
\declaretheorem[sibling=theorem]{lemma}
\declaretheorem[sibling=theorem]{corollary}

\declaretheorem[sibling=theorem,style=remark,qed=\qedsymbol]{remark}
\declaretheorem[sibling=theorem,style=remark,qed=\qedsymbol]{example}
\declaretheorem[sibling=theorem]{claim}
\declaretheoremstyle[
spaceabove=6pt,spacebelow=6pt,
headfont=\normalfont\bfseries,
notefont=\mdseries, notebraces={}{},
bodyfont=\normalfont,
postheadspace=1em,
numbered=no
]{problem}
\declaretheorem[name={Problem:},style=problem]{problem}
\usepackage{tikz}
\usetikzlibrary{fit}
\usetikzlibrary{positioning}
\usetikzlibrary{arrows}
\usetikzlibrary{automata}
\tikzstyle{state}=[circle,minimum size=.6cm,draw=gray!90,inner sep=0pt,fill=gray!20,thick]
\tikzstyle{every node}=[font=\footnotesize]
\tikzstyle{every edge}=[draw,>=stealth',shorten >=1pt,thin]
\newcommand{\F}[1]{\ComplexityFont{F}_{#1}}
\newcommand{\FGH}[1]{\mathscr{F}_{\!#1}}

\newcommand{\TOWER}{\ComplexityFont{TOWER}}
\newcommand{\ACK}{\ComplexityFont{ACKERMANN}}
\newcommand{\qout}[1]{q_\mathit{out}(#1)}
\newcommand{\qin}[1]{q_\mathit{in}(#1)}
\renewcommand{\exp}{\mathsf{exp}}

\newcommand{\setN}{\mathbb{N}}
\newcommand{\setZ}{\mathbb{Z}}

\newcommand{\size}[1]{| #1 |}

\newcommand{\norm}[2]{\|#1\|_{#2}}

\newcommand{\eqby}[1]{\stackrel{\textrm{{\normalfont\tiny{#1}}}}{=}}
\newcommand{\eqdef}{\eqby{def}}

\makeatletter
\newcommand{\facr}[1][\@empty]{\Pi_\rho\ifx#1\@empty\relax\else(#1)\fi}
\newcommand{\vect}[1][\@empty]{\vec
  V_{\!\rho}\ifx#1\@empty\relax\else(#1)\fi} \makeatother

\makeatletter
\newcommand{\step}[2][\@empty]{%
  \mathbin{\raisebox{0pt}[1em][.4ex]{\ifx#1\@empty%
    $\xrightarrow{\raisebox{-1.5pt}[0pt][.3ex]{\scriptsize\ensuremath{#2}}}$%
  \else%
    $\xrightarrow[{\raisebox{2pt}[0pt][0pt]{\scriptsize\ensuremath{#1}}}]%
                {\raisebox{-1.5pt}[0pt][.3ex]{\scriptsize\ensuremath{#2}}}$%
  \fi}}%
}%
\newcommand{\VS}[1][G]{\vec{V}_{\!\!#1}}
\makeatother

\input{dash}
\def\vec#1{\mathchoice{\mbox{\boldmath$\displaystyle#1$}}
{\mbox{\boldmath$\textstyle#1$}}
{\mbox{\boldmath$\scriptstyle#1$}}
{\mbox{\boldmath$\scriptscriptstyle#1$}}}
\newcommand{\rank}[1]{\operatorname{rank}(#1)}
\ifsubmission
  \newcommand{\appref}[1]{Appendix~\ref{#1}}
\else
  \newcommand{\appref}[1]{the full paper}
\fi
\providecommand{\qedhere}{\qed}

\renewcommand{\eqby}[1]{\mathrel{\raisebox{-.1ex}{\ensuremath{\stackrel{\raisebox{-.25ex}{\scalebox{.5}{\upshape\textrm{#1}}}}{=}}}}}
\usepackage{url,hyperref}

\providecommand{\urlstyle}[1]{}
\providecommand{\doi}[1]{\href{http://dx.doi.org/#1}{\nolinkurl{doi:#1}}}
\usepackage{cleveref}
\crefname{section}{Sec.}{Sections}
\Crefname{section}{Section}{Sections}
\crefname{subsection}{Sec.}{Sections}
\Crefname{subsection}{Section}{Sections}
\crefname{subsubsection}{Sec.}{Sections}
\Crefname{subsubsection}{Section}{Sections}
\crefname{theorem}{Thm.}{theorems}
\Crefname{theorem}{Theorem}{Theorems}
\crefname{lemma}{Lem.}{lemmata}
\Crefname{lemma}{Lemma}{Lemmata}
\crefname{fact}{Fact}{facts}
\Crefname{fact}{Fact}{Facts}
\crefname{corollary}{Cor.}{corollaries}
\Crefname{corollary}{Corollary}{Corollaries}
\crefname{proposition}{Prop.}{propositions}
\Crefname{proposition}{Proposition}{Propositions}
\crefname{claim}{Claim}{claims}
\Crefname{claim}{Claim}{Claims}
\crefname{definition}{Def.}{definitions}
\Crefname{definition}{Definition}{Definitions}
\crefname{example}{Ex.}{examples}
\Crefname{example}{Example}{Examples}
\crefname{remark}{Rmk.}{remarks}
\Crefname{remark}{Remark}{Remarks}
\crefname{figure}{Fig.}{figures}
\Crefname{figure}{Figure}{Figures}
\crefname{table}{Tab.}{tables}
\Crefname{table}{Table}{Tables}
\crefname{property}{Pty.}{properties}
\Crefname{property}{Property}{Properties}
\def\keywords{\smallskip\noindent\textsc{Keywords.}}

\begin{document}
\title[VAS Reachability is Primitive-Recursive in Fixed Dimension]{Reachability \mbox{in Vector Addition Systems} is Primitive-Recursive in Fixed Dimension}
\author[J.~Leroux and S.~Schmitz]{J\'er\^ome Leroux$^1$ and Sylvain Schmitz$^{2,3}$}
\address{$^1$~LaBRI, CNRS, Universit\'e de Bordeaux\\France}
\address{$^2$~LSV, ENS Paris-Saclay \& CNRS
\\Universit\'e Paris-Saclay
\\France
}
\address{$^3$~IUF\\France}
\begin{abstract}
  The reachability problem in vector addition systems is a central
question, not only for the static verification of these systems, but
also for many inter-reducible decision problems occurring in various
fields.  The currently best known upper bound on this problem is not
primitive-recursive, even when considering systems of fixed dimension.
We provide significant refinements to the classical decomposition
algorithm of \citeauthor{mayr81}, \citeauthor{kosaraju82}, and
\citeauthor{lambert92} and to its termination proof, which yield an
\ifams{\smaller\textsf{ACKERMANN}}\else{\smaller\textsf{\textbf{ACKERMANN}}}\fi\ upper bound in the general case,
and primitive-recursive upper bounds in fixed dimension.  While this
does not match the currently best known
\ifams{\smaller\textsf{TOWER}}\else{\smaller\textsf{\textbf{TOWER}}}\fi\ lower bound for reachability, it is
optimal for related problems.

  \keywords{}Vector addition system, Petri net, reachability, fast-growing complexity
\end{abstract}
\maketitle
\section{Introduction}%
\label{sec-intro}
\paragraph{Vector addition systems with states\nopunct}
(VASS) are basically finite state systems with vectors of integers as
transition weights, as depicted in \cref{fig-vass}.
\begin{figure}[htbp]
  \centering
  \begin{tikzpicture}[auto,on grid,node distance=1.6,initial text={}]
    \node[state,initial by arrow](p){$q_\mathit{in}$};
    \node[state,below left=of p](p1){$p$};
    \node[state,right=2.5 of p,accepting by arrow](q){$q_\mathit{out}$};
    \node[state,below right=of q](q1){$q$};
    \path[->,every node/.style={font=\scriptsize,inner sep=1pt}]
      (p)  edge[loop above] node {$\vec{a}_1{=}(0,2,0)$}  ()
      (p)  edge[bend left=10]  node {$\vec{a}_3{=}(1,0,0)$}  (q)
      (p)  edge[bend right=10,swap,inner sep=2pt] node {$\vec{a}_4{=}(1,0,-2)$} (q)
      (p)  edge[bend left,inner sep=0pt,near end]  node {$\vec{a}_2{=}(2,2,-1)$} (p1)
      (p1) edge[bend left, near start]  node {$\vec{a}_5{=}(1,0,-2)$} (p)
      (q)  edge[loop above] node {$\vec{a}_6{=}(1,-1,0)$} ()
      (q)  edge[bend left, near end]  node {$\vec{a}_7{=}(1,-1,-2)$} (q1)
      (q1) edge[bend left,inner sep=0pt,near start]  node {$\vec{a}_9{=}(0,0,0)$}  (q)
      (q1) edge[loop below,inner sep=3pt] node {$\vec{a}_8{=}(-2,-1,0)$}();
  \end{tikzpicture}
  \caption{\label{fig-vass}A vector addition system with states.}
\end{figure}
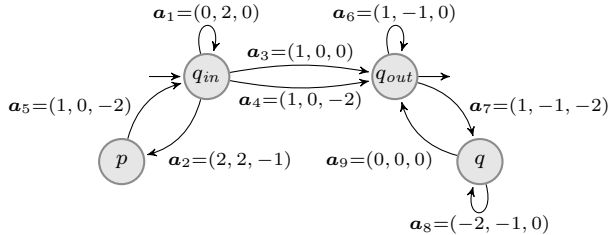
\noindent
Their semantics, starting from an initial vector
of natural numbers, simply adds component-wise the weights of the
successive transitions, but the current values should remain
non-negative at all times on every coordinate.  For instance, in the
three-dimensional system of \cref{fig-vass},\vspace*{-1ex}
\begin{multline*}
  q_\mathit{in}(0,0,2)\step{\vec a_1}q_\mathit{in}(0,2,2)\step{\vec a_1}q_\mathit{in}(0,4,2)\step{\vec a_3}q_\mathit{out}(1,4,2)\\\step{\vec a_6}q_\mathit{out}(2,3,2)\step{\vec a_7}
  q(3,2,0)\step{\vec a_8} q(1,1,0)\step{\vec a_9} q_\mathit{out}(1,1,0)
\end{multline*}
is a path witnessing that $q_\mathit{out}(1,1,0)$ can be reached
from $q_\mathit{in}(0,0,2)$, but for instance $q(1,1,0)\step{\vec a_8} q(-1,0,0)$
is not a valid execution step due to the negative value in the first
coordinate.

Vector addition systems with states are equivalent to Petri nets, and
well-suited whenever one needs to model discrete resources, for
instance threads in concurrent computations, molecules in chemical
reactions, organisms in biological processes,
etc.  %
They are also a crucial ingredient in many algorithms.  In particular,
the decidability of their \emph{reachability problem}
\citep{mayr81,kosaraju82,lambert92,leroux12} is the cornerstone of
many decidability results---see for instance \cite[Sec.~5]{schmitz16}
for a large sample of problems inter-reducible with VASS reachability
in logic, formal languages, verification, etc.

\medskip
In spite of its relevance to a wide range of problems, the complexity
of the VASS reachability problem is still not well understood.
Indeed, it turns out that this seemingly simple problem is both
conceptually and computationally very complex.

\paragraph{On a conceptual level,\nopunct} the \citeyear{mayr81} decidability proof
by \citet{mayr81} was the culmination of more than a decade of
research in the topic and is considered as one of the great
achievements of theoretical computer science.
Both \citeauthor{mayr81}'s \emph{decomposition algorithm} and its
proof are however quite intricate.  \Citet{kosaraju82} and
\citet{lambert92} contributed several simplifications
of \citet{mayr81}'s original arguments and \citet{leroux15} recast the
decomposition algorithm in a more abstract framework based on
well-quasi-order ideals, while \citet{leroux12} provides a very simple
algorithm with a short but non constructive proof, but none of these
developments can be called `easy' and the problem seems inherently
involved.

\paragraph{On a computational level,\nopunct} on the one hand, the
best known lower bound---which was from~\citeyear{lipton76} until very
recently \EXPSPACE-hardness~\citep{lipton76}---is now
\TOWER-hardness~\citep{czerwinski19}.  This new lower bound puts the
problem firmly in the realm of non-elementary complexity.  In this
realm, complexity is measured using the `fast-growing' complexity
classes $(\F\alpha)_\alpha$ from~\citep{schmitz13}, which form a
strict hierarchy indexed by ordinals.  The already mentioned
$\TOWER=\F3$ corresponds to problems solvable in time bounded by a
tower of exponentials; each $\F k$ for a finite~$k$ is primitive
recursive, and $\ACK=\F\omega$ corresponds to problems solvable with
Ackermannian resources\ifams\ (see \cref{fig-fg})\fi.  On the other hand, due to
the intricacy of the decomposition algorithm, it eluded analysis for a
long time until a `cubic Ackermann' upper bound was obtained
in~\citep{leroux15} at level $\F{\omega^3}$, with a slightly improved
$\F{\omega^2}$ upper bound in~\citep{schmitz17}.

\ifams\begin{figure}[hbtp]
  \centering\scalebox{.87}{
  \begin{tikzpicture}[every node/.style={font=\small}]
    \shadedraw[color=black!90,top color=black!20,middle
    color=black!5,opacity=20,shading angle=-15](-1,0) arc (180:0:4.8cm);
    \draw[color=blue!90,thick,fill=blue!20](-.7,0) arc (180:0:3.8cm);
    \draw[color=blue!90,fill=blue!20,thick](-.65,0) arc (180:0:3.5cm);
    \draw[color=violet!90!black,fill=violet!20,thick](-.6,0) arc (180:0:3.25cm);
    \shadedraw[color=black!90,top color=black!20,middle color=black!5,opacity=20,shading angle=-15](-.5,0) arc (180:0:3cm);
    \draw[color=blue!90,fill=blue!20,thick](-.1,0) arc (180:0:1.7cm);
    \shadedraw[color=black!90,top color=black!20,middle
    color=black!5,opacity=20,shading angle=-15,thin](0,0) arc (180:0:1.5cm);
    \draw (1.5,.5) node{$\ComplexityFont{ELEMENTARY}$};
    \draw (4,1.2) node[color=blue]{$\F3=\!\ComplexityFont{TOWER}$};
    \draw[color=blue!90,thick] (3.15,1) -- (3.05,.9);
    \draw (2.5,1.9) node{$\bigcup_k\!\F{k}{=}\ComplexityFont{PRIMITIVE\text-RECURSIVE}$};
    \draw (5.32,1.5) node[color=violet!90!black]{$\F\omega$};
    \draw (5.73,1.6) node[color=blue]{$\F{\!\omega^{\!2}}$};
    \draw (6.21,1.7) node[color=blue]{$\F{\!\omega^3}$};
    \draw (3.95,4) node{$\bigcup_k\!\F{\omega^{k}}=\ComplexityFont{MULTIPLY\text-RECURSIVE}$};
    \draw[loosely
    dotted,very thick,color=black!70](6.7,1.8) --
    (7.2,1.92); \end{tikzpicture}} \caption{Pinpointing
    $\F{\omega}=\ComplexityFont{ACKERMANN}$ among the complexity
    classes beyond \ComplexityFont{ELEMENTARY}~\citep{schmitz13}.\label{fig-fg}}
\end{figure}
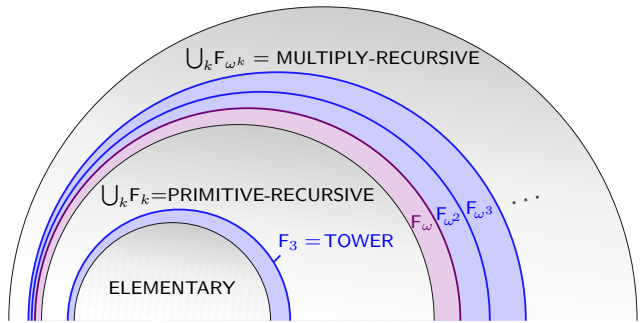\fi

This leaves a gigantic gap between the known lower and upper bounds.
This is however mitigated by the fact that the decomposition
algorithm, on which the upper bounds were obtained, provably has a non
primitive-recursive complexity.  This was already observed
by~\citet{muller85}, due to the algorithm's reliance on
\citeauthor{karp69} trees~\citep{karp69}.  Moreover, the full
decomposition produced by the algorithm contains more information than
just the existence of a reachability witness (which exists if and only
if the full decomposition is not empty).  For instance,
\citet{lambert92} exploits the full decomposition to derive a pumping
lemma for labelled VASS languages, \citet{habermehl10} further show
that one can compute a finite-state automaton recognising the
downward-closure of a labelled VASS language with respect to the
scattered subword ordering, and \citet{czerwinski18} show how to
exploit the decomposition for deciding language boundedness
properties.  In particular, the result of \citeauthor{habermehl10}
means that one can decide, given two labelled VASS, whether an
inclusion holds between the downward-closures of their languages,
which is an \ACK-hard problem~\cite{zetzsche16}.  Thus any algorithm
that returns such a full decomposition must be non
primitive-recursive.

\ifams\relax\else\medskip\fi
\paragraph{Contributions.\nopunct}In this paper, we show that VASS
reachability is in $\ACK=\F\omega$, and more precisely in
$\F{d+4}$ when the dimension~$d$ of the
system is fixed.  This improvement over the bound $\F{\omega^2}$
(resp.\ $\F{\omega\cdot(d+1)}$ in fixed dimension) shown
in~\citep{schmitz17} is obtained by analysing a decomposition
algorithm similar to those of \citet{mayr81}, \citet{kosaraju82}, and
\citet{lambert92}.  In a nutshell, a decomposition algorithm defines
both
{\ifams\relax\else\setlength\itemsep{0pt}\setlength\topsep{-5pt}\setlength\parskip{0pt}\setlength\partopsep{0pt}\fi
\begin{itemize}
\item a structure (resp.\ `regular constraint graphs' for
  \citeauthor{mayr81}, `generalised VASSes' for
  \citeauthor{kosaraju82}, and `marked graph-transition sequences' for
  \citeauthor{lambert92})---see \cref{sec-chara}---and
\item a condition on this structure that ensures there is an
  execution witnessing reachability (resp.\ `consistent marking',
  `property~$\uptheta$', and `perfectness')---see \cref{sub-normalnonempty}.
\end{itemize}}
The algorithms compute a decomposition by successive refinements of
the structure until the condition is fulfilled, by which time the
existence of an execution becomes guaranteed---see \cref{sec-klm}.

We work in this paper with a decomposition algorithm quite similar to
that of \citet{kosaraju82}, for which the reader will find good
expositions for instance in~\citep{muller85,reutenauer90,lasota18}.
We benefit however from two key insights (which in turn require
significant adaptations throughout the algorithm).

The first key insight is a new termination argument for the
decomposition process, based on the dimensions of the vector spaces
spanned by the cycles of the structure (see \cref{sub-rank}).  On its
own, this new termination argument would already be enough to yield
$\ACK$ upper bounds and primitive-recursive ones in fixed
dimension.

The second key insight lies within the decomposition process itself,
where we show using techniques inspired by \citet{rackoff78} that we
can eschew the computation of \citeauthor{karp69}'s coverability
trees, and therefore the worst-case Ackermannian blow-up that arises
from their use~\citep{cardoza76}---see \cref{sub-pump}.  In
itself, this new decomposition algorithm would not bring the
complexity below the previous bounds, but combined with the first
insight, it yields rather tight upper bounds, at level
$\F{d+4}$ in fixed dimension~$d$---see \cref{sec-up}.

In fact, the new upper bounds apply to other decision problems.  As we
discuss in \cref{sec-apps}, \citeauthor{zetzsche16}'s $\ACK$ lower
bound~\citep{zetzsche16} can be refined to prove that the inclusion
problem between the downward-closures of two labelled VASS languages
is $\F d$-hard in fixed dimension $d\geq 3$, thus close to matching the
$\F{d+4}$ upper bound one obtains by applying the results of
\citet{habermehl10} to our decomposition algorithm.

We start in \cref{sec-vass} by recalling basic definitions and
notations on vector addition systems.
The full proofs for the decomposition
algorithm are presented in Appendices~\ref{sec-pump}
to~\ref{sec-normal}.

\section{Background}
\label{sec-vass}
\paragraph{Notations}

Let $\+N_\omega\eqdef\+N\uplus\{\omega\}$ extend the set of natural
numbers with an infinite element~$\omega$ with $n<\omega$ for all
$n\in\+N$.  We also use the partial order $\sqsubseteq$ over
$\setN_\omega$ defined by $x\sqsubseteq y$ if $y\in\{x,\omega\}$.

Let $d\in\+N$ be a dimension.  The relations~$\leq$ and~$\sqsubseteq$
are extended component-wise to vectors in $\+N_\omega^d$.  The
components of a vector that are equal to~$\omega$ intuitively denote
arbitrarily large values; we call a vector in~$\+N^d$ \emph{finite}.
Given a vector $\vec{x}\in\setN_\omega^d$ and a subset
$I\subseteq \{1,\ldots,d\}$ of the components, we denote by
$\vec{x}|_I$ the vector obtained from $\vec{x}$ by replacing
components \emph{not} in $I$ by~$\omega$.  Note that
$\vec x\sqsubseteq\vec y$ implies $\vec x\leq\vec y$ and that
$\vec x\sqsubseteq \vec x|_I$ for all $\vec x,\vec y\in\+N^d_\omega$
and $I\subseteq\{1,\dots,d\}$.  For instance, for $d=3$,
$(3,2,1)\leq (4,\omega,1)$ but $(3,2,1)\not\sqsubseteq (4,\omega,1)$;
if $I=\{2,3\}$, then $(3,2,1)|_I=(\omega,2,1)$ and
$(4,\omega,1)|_I=(\omega,\omega,1)$, and then
$(\omega,2,1)\sqsubseteq(\omega,\omega,1)$.  We let $\vec 0$ denote
the zero vector and $\vec\omega$ the vector with
$\vec\omega(i)\eqdef\omega$ for all $1\leq i\leq d$.  Observe that
$\vec x\sqsubseteq\vec\omega$ for all $\vec x\in\+N_\omega^d$.

For a vector $\vec x\in\+N^d_\omega$, its norm $\|\vec x\|$ is defined
over its finite components as
$\sum_{1\leq i\leq d\mid \vec x(i)<\omega}\vec x(i)$ \ifams (a sum over an
empty set is zero)\fi; for a vector
$\vec x\in\+Z^d$, we let as usual
$\|\vec x\|\eqdef\sum_{1\leq i\leq d}|\vec x(i)|$.  For instance,
$\|(3,\omega,1)\|=4$ and $\|(-4,2,1)\|=7$.

\ifams\relax\else\smallskip\fi
\paragraph{Vector Addition Systems}
While we focus in this paper on reachability in vector addition
systems with a finite set of control states, we also rely on notations
for the simpler case of vector addition systems.

A \emph{vector addition system} (VAS)~\citep{karp69} of
dimension~$d\in\+N$ is a finite set $\vec{A}\subseteq \setZ^d$ of
vectors called \emph{actions}.  The semantics of a VAS is defined over
\emph{configurations} in~$\+N_\omega^d$.  We associate to an action
$\vec{a}\in\vec{A}$ the binary relation $\step{\vec{a}}$ over
configurations by $\vec{x}\step{\vec{a}}\vec{y}$ if
$\vec{y}=\vec{x}+\vec{a}$, where addition is performed component-wise
with the convention that $\omega+z=\omega$ for every $z\in\setZ$.
Given a finite word $\sigma=\vec{a}_1\ldots\vec{a}_k\in\vec A^\ast$ of
actions we also define the binary relation $\step{\sigma}$ over
configurations by $\vec{x}\step{\sigma}\vec{y}$ if there exists a
sequence $\vec{c}_0,\ldots,\vec{c}_k$ of configurations such that
$$\vec{x}=\vec{c}_0\step{\vec{a}_1}\vec{c}_1\cdots \step{\vec{a}_k}\vec{c}_k=\vec{y}\;.$$

The \emph{VAS reachability problem} consists in deciding given two
finite configurations $\vec{c}_\mathit{in},\vec{c}_\mathit{out}$ in
$\+N^d$ and a VAS $\vec{A}$ whether there exists a word
$\sigma\in\vec{A}^*$ such that
$\vec{c}_\mathit{in}\step{\sigma}\vec{c}_\mathit{out}$.

\ifams\relax\else\smallskip\fi
\paragraph{Vector Addition Systems with States}

A \emph{vector addition system with states} (VASS)~\citep{hopcroft79}
of dimension~$d\in\+N$ is a triple $G=(Q,q_\mathit{in},q_\mathit{out},
T)$ where~$Q$ is a non-empty finite set of states, $q_\mathit{in}\in
Q$ is the input state, $q_\mathit{out}\in Q$ is the output state,
and~$T$ is a finite set of transitions in $Q\times \setZ^d\times Q$; $\vec
A\eqdef\{\vec a\mid \exists p,q\in Q\mathbin.(p,\vec a,q)\in T\}$ is the
associated set of actions.

\begin{example}\label{ex-vass}
  \Cref{fig-vass} depicts the VASS
  $G_\mathrm{ex}=(Q_\mathrm{ex},q_\mathit{in},q_\mathit{out},
  T_\mathrm{ex})$ of dimension~$3$ where
  $Q_\mathrm{ex}=\{q_\mathit{in},q_\mathit{out},p,q\}$ and
  $T_\mathrm{ex}=\{t_1,t_2,t_3,t_4,t_5,t_6,t_7,t_8,t_9\}$
  with {\small\begin{align*} t_1&=(q_\mathit{in},(0,2,0),q_\mathit{in})\;,&
  t_2&=(q_\mathit{in},(2,2,-1),p)\;,\\
  t_3&=(q_\mathit{in},(1,0,0),q_\mathit{out})\;,&
  t_4&=(q_\mathit{in},(1,0,-2),q_\mathit{out})\;,\\
  t_5&=(p,(1,0,-2),q_\mathit{in})\;,&
  t_6&=(q_\mathit{out},(1,-1,0),q_\mathit{out})\;,\\
  t_7&=(q_\mathit{out},(1,-1,-2),q)\;,& t_8&=(q,(-2,-1,0),q)\;,\\
  t_9&=(q,(0,0,0),q_\mathit{out})\;.&&\qedhere  \end{align*}}
\end{example}

We focus on VASSes in this paper rather than VASes, because we exploit
the properties of their underlying directed graphs.  A
\emph{path}~$\pi$ in a VASS~$G$ from a state~$p$ to a state~$q$
labelled by a word $\vec{a}_1\ldots\vec{a}_k$ of actions is a word of
transitions of $G$ of the form
$(p_1,\vec{a}_1,q_1)\ldots (p_{k},\vec{a}_k,q_k)$ with $p_0=p$,
$q_k=q$, and $q_j=p_{j+1}$ for all $1\leq j<k$.  Such a path is
\emph{complete} if $p=q_\mathit{in}$ and $q=q_\mathit{out}$ are
the input and output states of~$G$.  A \emph{cycle} on a
state $q$ is a path from $q$ to~$q$.
\ifams\begin{example}\label{ex-path}\fi
For instance, in \cref{ex-vass},
the execution presented in the introduction corresponds to the path
$\pi_\mathrm{ex}=t_1\,t_1\,t_3\,t_6\,t_7\,t_8\,t_9$ labelled by
$\sigma_\mathrm{ex}=\vec a_1\,\vec a_1\,\vec a_3\,\vec a_6\,\vec
a_7\,\vec a_8\,\vec a_9$, and is complete.
\ifams\end{example}\fi

We write $p\equiv_G q$ if there exists a path from $p$ to $q$ and a
path from $q$ to $p$; this defines an equivalence relation whose
equivalence classes are called the \emph{strongly connected
components} of~$G$.  In \cref{ex-vass}, the strongly connected
components are $\{q_\mathit{in},p\}$ and~$\{q,q_\mathit{out}\}$.  A
VASS $G=(Q,q_\mathit{in},q_\mathit{out},T)$ is said to
be \emph{strongly connected} if $Q$ is a strongly connected component
of $G$.

The \emph{Parikh image} of a path $\pi$ is the function $\phi{:}\,T\to\+N$
that maps each transition~$t\in T$ to its number of occurrences
in~$\pi$.  The \emph{displacement} of a path~$\pi$ labelled by a word
$\vec{a}_1\ldots\vec{a}_k$ of actions is the vector
$\Delta(\pi)\eqdef\sum_{j=1}^k\vec{a}_j$; note that this is equal to
$\Delta(\phi)\eqdef\sum_{t=(p,\vec a,q)\in T}\phi(t)\cdot\vec a$
if~$\phi$ is the Parikh image of~$\pi$.
\ifams\begin{example}\label{ex-parikh}\fi
For the example path~$\pi_\mathrm{ex}$ \ifams from \cref{ex-path}\fi,
$\phi_\mathrm{ex}=(2,0,1,0,0,1,1,1,1)$ and
$\Delta(\pi_\mathrm{ex})=(1,1,-2)$.
\ifams\end{example}\fi

A \emph{state-configuration} of a VASS
$G=(Q,q_\mathit{in},q_\mathit{out},T)$ is a pair $(q,\vec{x})\in
Q\times\+N_\omega^d$ denoted by $q(\vec{x})$ in the sequel.  Given an
action $\vec{a}$ we define the step relation $\step[G]{\vec a}$ over
state-configurations by $p(\vec{x}) \step[G]{\vec{a}} q(\vec{y})$ if
$(p,\vec{a},q)\in T$ and $\vec{x}\step{\vec{a}}\vec{y}$.  By
extension, given a word $\sigma$ of actions
$\sigma=\vec{a}_1\ldots\vec{a}_k$, $p(\vec{x}) \step[G]\sigma
q(\vec{y})$ if there exists a sequence
$q_0(\vec{c}_0),\ldots,q_k(\vec{c}_k)$ of state-configurations such
that
$$p(\vec{x})=q_0(\vec{c}_0)\step[G]{\vec{a}_1}q_1(\vec{c}_1)\cdots\step[G]{\vec{a}_k}q_k(\vec{c}_k)=q(\vec{y})\;.$$
Notice that $p(\vec{x})\step[G]\sigma q(\vec{y})$ if, and only if,
there exists a path in~$G$ from~$p$ to~$q$ labelled by~$\sigma$ such
that $\vec{x}\step{\sigma}\vec{y}$.  In \cref{ex-vass},
$\qin{(0,0,2)}\step[G_\mathrm{ex}]{\sigma_\mathrm{ex}}\qout{(1,1,0)}$.
Finally, we write $p(\vec{x})\step[G]\ast q(\vec{y})$ if there exists
$\sigma\in\vec A^\ast$ such that $p(\vec{x})\step[G]\sigma q(\vec{y})$.

\ifams\relax\else\smallskip\fi
\paragraph{Reachability}  We focus in this paper on the
following decision problem.

\vspace*{-.5ex}
\begin{problem}[VASS reachability]\label{reach}
\hfill\vspace*{-.5ex}\begin{description}[\IEEEsetlabelwidth{question}]
\item[input] a VASS $G=(Q,q_\mathit{in},q_\mathit{out},T)$ of
dimension~$d$ and two finite configurations $\vec c_\mathit{in},\vec
c_\mathit{out}\in\+N^d$
\item[question] does $\qin{\vec c_\mathit{in}}\step[G]{*}\qout{\vec
c_\mathit{out}}$ hold?
\end{description}
\end{problem}
The previously mentioned VAS reachability problem reduces to
\nameref{reach}: given a VAS~$\vec A$ and two finite configurations
$\vec c_\mathit{in},\vec c_\mathit{out}$, it suffices to consider the
\nameref{reach} problem with input
$(\{q\},q,q,\{q\}\times\vec A\times\{q\})$ and the same configurations
$\vec c_\mathit{in},\vec c_\mathit{out}$.  A converse reduction is
possible by encoding the states, at the expense of increasing the
dimension by three~\citep{hopcroft79}.

\section{Decomposition Structures}
\label{sec-chara}
The version of the decomposition algorithm we present in
\cref{sec-klm} proceeds globally as the ones of \citeauthor{mayr81},
\citeauthor{kosaraju82}, and \citeauthor{lambert92}, and we call the
underlying structures \emph{KLM sequences} after them.

\subsection{KLM Sequences}
A \emph{KLM sequence} $\xi$ of dimension~$d$ is a sequence
\begin{equation}
  \xi= (\vec{x}_0G_0\vec{y}_0)\vec{a}_1 (\vec{x}_1G_1\vec{y}_1)\ldots
  \vec{a}_k(\vec{x}_kG_k\vec{y}_k)\label{eq-klm}
\end{equation}
where $\vec{x}_0,\vec{y}_0,\ldots,\vec{x}_k,\vec{y}_k$ are
configurations, $G_0,\ldots,G_k$ are VASSes of dimension~$d$, and
$\vec{a}_1,\ldots,\vec{a}_k$ are actions.  KLM sequences are
essentially the same as \citeauthor{kosaraju82}'s `generalised
VASSes'~\citep{kosaraju82}, except that we do not require
$G_0,\dots,G_k$ to be strongly connected.

The \emph{action language} of a KLM sequence $\xi$ is the set $L_\xi$
of words of actions of the form
$\sigma_0\vec{a}_1\sigma_1\ldots\vec{a}_k\sigma_k$ such
that~$\sigma_j$ is the label of a complete path of~$G_j$ for
every~$j$, and such that there exists a sequence
$\vec{m}_0,\vec{n}_0,\ldots,\vec{m}_k,\vec{n}_k$ of configurations in
$\setN^d$ such that
\begin{equation}\label{eq-lang}
  \vec{m}_0\xrightarrow{\sigma_0}\vec{n}_0\xrightarrow{\vec{a}_1}\cdots
\vec{m}_k\xrightarrow{\sigma_k}\vec{n}_k\,\end{equation}
where $\vec{m}_j\sqsubseteq \vec{x}_j$ and
$\vec{n}_j\sqsubseteq \vec{y}_j$ for every $0\leq j\leq k$.

Note that the reachability problem for a VASS~$G$ and two finite
configurations $\vec{c}_{in},\vec{c}_{out}\in\setN^d$ reduces to the
non-emptiness of the action language of the KLM sequence
$(\vec{c}_{in}G\vec{c}_{out})$.  In fact, in that case, the action
language is the set of words $\sigma\in\vec A^\ast$ such that
$\qin{\vec{c}_{in}}\step[G]{\sigma}\qout{\vec{c}_{out}}$.

\begin{example}\label{ex-lang}
  In \cref{ex-vass},
  $\xi_\mathrm{ex}=((0,0,2)G_\mathrm{ex}(1,1,0))$ is a KLM sequence
  with action language
  \begin{align*}
    L_{\xi_\mathrm{ex}}&=\{\vec a_1^{2+3n}\vec a_3\,\vec a_6^{1+4n}\vec
  a_7\,\vec a_8^{1+2n}\vec a_9\mid n\in\+N\}\\
   &\:\cup\{\vec a_1^{2+3n}\vec a_3\,\vec a_6^{4n}\vec
  a_7\,\vec a_8^{1+2n}\vec a_9\,\vec a_6\mid n\in\+N\}\;.\qedhere
  \end{align*}
\end{example}

\subsection{Ranks and Sizes}\label{sub-rank}

\paragraph{Vector Spaces}

We associate to a transition $t$ of a VASS $G$ the vector space
$\VS(t)\subseteq\+Q^d$ spanned by the \ifams\relax\else\pagebreak\fi
displacements of the cycles that contain~$t$. The following lemma
shows that this vector space only depends on the strongly connected
components of~$G$.
\begin{lemma}\label{lem:cyclespan}
  Let $t$ be a transition of a strongly connected VASS $G$. Then the
  vector space $\VS(t)$ is equal to the vector space
  spanned by the displacements of the cycles of~$G$.
\end{lemma}
\begin{proof}
  Let $\vec{V}$ be the vector space spanned be the displacements of
  the cycles of~$G$. Naturally, we have $\VS(t)\subseteq
  \vec{V}\!$. For the converse, let us consider a sequence
  $\theta_1,\ldots,\theta_k$ of cycles such that $\theta_j$ is a cycle
  on a state $q_j$ for every $1\leq j\leq k$, and such that
  $\Delta(\theta_1),\ldots,\Delta(\theta_k)$ span the vector space
  $\vec{V}$. Since $G$ is strongly connected, there exists a path
  $\pi_j$ from $q_{j-1}$ to $q_{j}$ for every $j\in\{1,\ldots,k\}$
  with $q_0\eqdef q_k$. Moreover, we can assume without loss of
  generality that $t$ occurs in the cycle
  $\theta\eqdef\pi_1\ldots\pi_k$. Let $\theta_j'$ be the cycle
  obtained from $\theta$ by inserting $\theta_j$ in $q_j$ and formally
  defined as
  $\theta_j'\eqdef\pi_1\ldots\pi_j\theta_j\pi_{j+1}\ldots\pi_k$. Observe
  that $\Delta(\theta)$ and $\Delta(\theta_j')$ are both in
  $\VS(t)$ since $t$ occurs in the cycles $\theta$ and
  $\theta_j'$. As $\Delta(\theta_j)=\Delta(\theta_j')-\Delta(\theta)$,
  it follows that $\Delta(\theta_j)\in \VS(t)$. We derive
  that the vector space spanned by
  $\Delta(\theta_1),\ldots,\Delta(\theta_k)$ is included in
  $\VS(t)$. Hence $\vec{V}\subseteq\VS(t)$.
\end{proof}
As a corollary, if two transitions $t$ and $t'$ are induced by the
same strongly connected component of a VASS~$G$, then
$\VS(t)=\VS(t')$.

\ifams\relax\else\smallskip\fi
\paragraph{Ranks}
The \emph{rank} of a VASS $G$ is the tuple
$\rank{G}\eqdef(r_d,\ldots,r_0)\in\+N^{d+1}$ where $r_i$ is the number of
transitions $t\in T$ such that the dimension of $\VS(t)$ is equal
to~$i$. The \emph{rank} of a KLM sequence $\xi$ defined as
$(\vec{x}_0G_0\vec{y}_0)\vec{a}_1 (\vec{x}_1G_1\vec{y}_1)\ldots
\vec{a}_k(\vec{x}_kG_k\vec{y}_k)$
is the vector $\rank{\xi}\eqdef\sum_{j=0}^k\rank{G_j}$ where ranks are
added component-wise.

Ranks are ordered lexicographically by the relation
$\leq_\mathit{lex}$ defined by
$(r_d,\ldots,r_0)\leq_\mathit{lex}(s_d,\ldots,s_0)$ if they are equal
or if the minimal~$i$ such that $r_i\not=s_i$ satisfies~$r_i<s_i$.
Note that the linear order $(\+N^{d+1},<_\mathit{lex})$ is well-founded,
with order type~$\omega^{d+1}$.
In \citeauthor{kosaraju82}'s decomposition algorithm, the rank of a
KLM sequence was defined as a multiset of triples
$(n_{j,1},n_{j,2},n_{j,3})$ for all $0\leq j\leq k$, with
$n_{j,1}\leq d$, $n_{j,2}\eqdef|T_j|$, and $n_{j,3}\leq 2d$, where the
triples are ordered lexicographically and the multisets using
\citeauthor{DM1979} multiset ordering~\citep{DM1979}.  This ranking
function ranged over an order type in $\omega^{\omega^3}$, and
actually below~$\omega^{\omega\cdot(d+1)}$~\citep{schmitz17}.

\begin{example}
  In \cref{ex-vass},
  \ifams\begin{align*}
   \VS(t_3)&=\VS(t_4)=\{(0,0,0)\},\;\\
   \VS(t_1)&=\VS(t_2)=\VS(t_5)=\mathrm{span}((0,2,0),(3,2,-3))\;,\\
   \VS(t_6)&=\VS(t_7)=\VS(t_8)=\VS(t_9)=\mathrm{span}((-2,-1,0),(1,-1,-2),(1,-1,0))\;.
   \end{align*}\else
  \begin{align*}
   \VS(t_3)&=\VS(t_4)=\{(0,0,0)\},\;\\
   \VS(t_1)&=\VS(t_2)=\VS(t_5)\\&=\mathrm{span}((0,2,0),(3,2,-3))\;,\\
   \VS(t_6)&=\VS(t_7)=\VS(t_8)=\VS(t_9)\\&=\mathrm{span}((-2,-1,0),(1,-1,-2),(1,-1,0))\;.
   \end{align*}\fi Thus
  $\rank{G_\mathrm{ex}}=(4,3,0,2)=\rank{\xi_\mathrm{ex}}$.
\end{example}

\paragraph{Sizes}
The \emph{size} of a VASS
$G=(Q,q_\mathit{in},q_\mathit{out},T)$ is%
\begin{align}\label{eq-G-size}
\size{G}&\eqdef 
|Q|+|T|+\sum_{t\in T}\norm{\Delta(t)}{}\;.\intertext{The \emph{size} of a KLM
sequence $\xi$ of the form $(\vec{x}_0G_0\vec{y}_0)\vec{a}_1
\ldots
\vec{a}_k(\vec{x}_kG_k\vec{y}_k)$ is the natural number}
\size{\xi}&\eqdef \text{\small$2(d+1)^{d+1}\big(k+\sum_{j=1}^k\norm{\vec{a}_j}{}+\sum_{j=0}^k(\norm{\vec{x}_j}{}+\size{G_j}+\norm{\vec{y}_j}{})\big)$}.\label{eq-klm-size}
\end{align}

\subsection{Characteristic System}\label{sub-chara}
The action language of a KLM sequence can be over-approximated thanks
to a system of linear equations called its \emph{characteristic
  system}, which we are about to define.
Let us first associate to a VASS
$G=(Q,q_\mathit{in},q_\mathit{out},T)$ a \emph{Kirchhoff system} $K_G$
of linear equations such that 
$\phi\in\setN^T$ is a
model of $K_G$ if, and only if, the following constraint holds
\begin{equation}\label{eq-kir}
  \one_{q_\mathit{out}}\!-\one_{q_\mathit{in}}=\!\!\!\!\sum_{t=(p,\vec a,q)\in
    T}\!\!\!\!\!\phi(t)(\one_q-\one_p)\,,
\end{equation}
where
$\one_q{:}\,Q\rightarrow\{0,1\}$ is the \emph{characteristic
    function} of $q\in Q$ defined by $\one_q(p)\eqdef 1$ if $p=q$ and
  $\one_q(p)\eqdef 0$ otherwise.
Let us observe that the Parikh image of a path from $q_\mathit{in}$ to
$q_\mathit{out}$ in $G$ is a model of $K_G$.

A \emph{characteristic sequence} of a KLM
sequence of the form~$\xi=(\vec{x}_0G_0,\vec{y}_0)\vec{a}_1
\ifams\ldots$ $\else(\vec{x}_1G_1\vec{y}_1)\ldots\fi\vec{a}_k(\vec{x}_k,G_k,\vec{y}_k)$
where~$T_j$ is the set of transitions of~$G_j$ for each~$j$
is a sequence $\vec{h}=(\vec{m}_j,\phi_j,\vec{n}_j)_{0\leq j\leq k}$
of triples
$(\vec{m}_j,\phi_j,\vec{n}_j)\in\setN^d\times\setN^{T_j}\times\setN^d$. We
denote by $\norm{\vec{h}}{}$ the value
$\sum_{j=0}^k\big(\norm{\vec{m}_j}{}+\sum_{t\in
  T_j}\phi_j(t)+\norm{\vec{n}_j}{}\big)$.
We also denote by
$(\vec{m}_{j}^{\vec{h}},\phi_{j}^{\vec{h}},\vec{n}_{j}^{\vec{h}})$ the~$j$th
triple $(\vec{m}_j,\phi_j,\vec{n}_j)$ occurring in $\vec{h}$.

The \emph{characteristic system} of~$\xi$
is the system~$E_\xi$ of linear equations such that a characteristic
sequence $\vec{h}=(\vec{m}_j,\phi_j,\vec{n}_j)_{0\leq j\leq k}$ is a model
of $E_\xi$ if, and only if, the following two conditions hold:
\begin{enumerate}
\item $\vec{m}_j^{\vec{h}}\sqsubseteq \vec{x}_j$,
  $\phi_j^{\vec{h}}\models K_{G_j}$,
  $\vec{n}_j^{\vec{h}}=\vec{m}_j^{\vec{h}}+\Delta(\phi_j^{\vec{h}})$,
  and $\vec{n}_j^{\vec{h}}\sqsubseteq \vec{y}_j$ for every
  $0\leq j\leq k$, and
\item $\vec{n}^{\vec{h}}_{j-1}\step{\vec{a}_j}\vec{m}^{\vec{h}}_j$ for
  every $1\leq j\leq k$.
\end{enumerate}

A KLM sequence $\xi$ is said to be \emph{satisfiable} if its
characteristic system $E_\xi$ is satisfiable. It is said to be
\emph{unsatisfiable} otherwise.  

\begin{example}\label{ex-char}
  Let us consider $\xi_\mathrm{ex}=((0,0,2)G_\mathrm{ex}(1,1,0))$ from
  \cref{ex-vass}.  Its characteristic system is
   \ifams{\begin{align*}
    &&\vec m&=(0,0,2)\wedge\vec n=(1,1,0)\\
    &\wedge& \vec n(1)&=\vec m(1)+2\phi(t_2)+\phi(t_3)+\phi(t_4)+\phi(t_5)+\phi(t_6)+\phi(t_7)-2\phi(t_8)\\
    &\wedge&\vec n(2)&=\vec m(2)+2\phi(t_1)+2\phi(t_2)-\phi(t_6)-\phi(t_7)-\phi(t_8)\\
    &\wedge&\vec n(3)&=\vec m(3)-\phi(t_2)-2\phi(t_4)-2\phi(t_5)-2\phi(t_7)\\
    &\wedge&-1&=-\phi(t_2)-\phi(t_3)-\phi(t_4)+\phi(t_5)\\
    &\wedge&0&=\phi(t_2)-\phi(t_3)\\
    &\wedge&0&=\phi(t_7)-\phi(t_9)\\
    &\wedge&1&=\phi(t_3)+\phi(t_4)-\phi(t_7)+\phi(t_9)\;,
  \end{align*}}\else
   {\small\begin{align*}
    &&\vec m&=(0,0,2)\wedge\vec n=(1,1,0)\\
    &\wedge& \vec n(1)&=\vec m(1)+2\phi(t_2)+\phi(t_3)+\phi(t_4)+\phi(t_5)+\phi(t_6)\\&&&\quad+\phi(t_7)-2\phi(t_8)\\
    &\wedge&\vec n(2)&=\vec m(2)+2\phi(t_1)+2\phi(t_2)-\phi(t_6)-\phi(t_7)-\phi(t_8)\\
    &\wedge&\vec n(3)&=\vec m(3)-\phi(t_2)-2\phi(t_4)-2\phi(t_5)-2\phi(t_7)\\
    &\wedge&-1&=-\phi(t_2)-\phi(t_3)-\phi(t_4)+\phi(t_5)\\
    &\wedge&0&=\phi(t_2)-\phi(t_3)\\
    &\wedge&0&=\phi(t_7)-\phi(t_9)\\
    &\wedge&1&=\phi(t_3)+\phi(t_4)-\phi(t_7)+\phi(t_9)\;,
  \end{align*}}\fi where the last four equations correspond
   to~$K_{G_\mathrm{ex}}$.  One can check
   that~\ifams the tuple \fi$((0,0,2),\phi_\mathrm{ex},(1,1,0))$ is a
  model\ifams, where $\phi_\mathrm{ex}$ was defined in \cref{ex-parikh}\fi.%
\end{example}

\begin{lemma}\label{lem-unsat}
  The action language of an unsatisfiable KLM sequence is empty.
\end{lemma}
\begin{proof}
  Assume that $L_\xi$ contains a word $\sigma$, and let us prove that $E_\xi$
  is satisfiable. As $\sigma\in L_\xi$, there exists a decomposition of $\sigma$ into
  $\sigma_0\vec{a}_1\sigma_1\ldots \vec{a}_k\sigma_k$ such that 
  $\sigma_j$ is the label of a complete path $\pi_j$ of $G_j$, 
  and there exists a sequence 
  $\vec{m}_0,\vec{n}_0,\ldots,\vec{m}_k,\vec{n}_k$ of vectors in
  $\setN^d$ with $\vec{m}_j\sqsubseteq \vec{x}_j$ and
  $\vec{n}_j\sqsubseteq \vec{y}_j$ for every $0\leq j\leq k$, and such
  that
  $$\vec{m}_0\step{\sigma_0}\vec{n}_0\step{\vec{a}_1}\cdots \step{\vec{a}_k}
  \vec{m}_k\step{\sigma_k}\vec{n}_k\;.$$
  Let $\phi_j$ be the Parikh image of $\pi_j$; \ifams then \else
  observe that \fi the
  characteristic sequence
  $(\vec{m}_j,\phi_j,\vec{n}_j)_{0\leq j\leq k}$ is a model of~$E_\xi$.
\end{proof}

\subsection{Homogeneous Characteristic System}\label{sub-hom}
In the sequel, variables whose values are bounded by the
characteristic system will provide a way of decomposing KLM sequences.
Since~$E_\xi$ is a system of linear
equations, bounded variables are characterised thanks to the
homogeneous form~$E^0_\xi$ of~$E_\xi$, called the \emph{homogeneous
  characteristic system} of~$\xi$ that we are about to define.

First, we define the \emph{homogeneous form}~$K_G^0$ of the Kirchhoff
  system~$K_G$ as the system of linear equation such
that~$\phi\in\+N^T$ is a model of $K_G^0$ if, and only
if, the following constraint holds
\begin{gather}
  \sum_{t=(p,\vec a,q)\in
    T}\phi(t)(\one_q-\one_p)=0\;.
\end{gather}
The \emph{homogenerous characteristic system} $E_\xi^0$ is such that a
sequence
\ifams $(\vec{m}_0,\phi_0,\vec{n}_0),\ldots,$
$(\vec{m}_k,\phi_k,\vec{n}_k)$\else
$(\vec{m}_0,\phi_0,\vec{n}_0),\ldots,(\vec{m}_k,\phi_k,\vec{n}_k)$\fi\ of
triples
$(\vec{m}_j,\phi_j,\vec{n}_j)\in\setN^d\times\setN^{T_j}\times\setN^d$
is a model of~$E_\xi^0$ if, and only if, the following two conditions
hold:
\begin{enumerate}
\item $\bigwedge_{i\mid \vec{x}_j(i)\not=\omega}\vec{m}_j(i)=0$,
  $\phi_j\models K_{G_j}^0$, $\vec{n}_j=\vec{m}_j+\Delta(\phi_j)$, and
  $\bigwedge_{i\mid \vec{y}_j(i)\not=\omega}\vec{n}_j(i)=0$ for every $0\leq j\leq k$, and
\item $\vec{n}_{j-1}=\vec{m}_j$ for every $1\leq
  j\leq k$.
\end{enumerate}

\begin{example}\label{ex-homochar}
  Let us consider $\xi_\mathrm{ex}=((0,0,2)G_\mathrm{ex}(1,1,0))$ from
  \cref{ex-vass}.  Its homogeneous characteristic system is
  \ifams{\begin{align*}
    &&\vec m&=(0,0,0)\wedge\vec n=(0,0,0)\\
    &\wedge& \vec n(1)&=\vec m(1)+2\phi(t_2)+\phi(t_3)+\phi(t_4)+\phi(t_5)+\phi(t_6)+\phi(t_7)-2\phi(t_8)\\
    &\wedge&\vec n(2)&=\vec m(2)+2\phi(t_1)+2\phi(t_2)-\phi(t_6)-\phi(t_7)-\phi(t_8)\\
    &\wedge&\vec n(3)&=\vec m(3)-\phi(t_2)-2\phi(t_4)-2\phi(t_5)-2\phi(t_7)\\
    &\wedge&0&=-\phi(t_2)-\phi(t_3)-\phi(t_4)+\phi(t_5)\\
    &\wedge&0&=\phi(t_2)-\phi(t_3)\\
    &\wedge&0&=\phi(t_7)-\phi(t_9)\\
    &\wedge&0&=\phi(t_3)+\phi(t_4)-\phi(t_7)+\phi(t_9)\;,
  \end{align*}}\else{\small\begin{align*}
    &&\vec m&=(0,0,0)\wedge\vec n=(0,0,0)\\
    &\wedge& \vec n(1)&=\vec m(1)+2\phi(t_2)+\phi(t_3)+\phi(t_4)+\phi(t_5)+\phi(t_6)\\&&&\quad+\phi(t_7)-2\phi(t_8)\\
    &\wedge&\vec n(2)&=\vec m(2)+2\phi(t_1)+2\phi(t_2)-\phi(t_6)-\phi(t_7)-\phi(t_8)\\
    &\wedge&\vec n(3)&=\vec m(3)-\phi(t_2)-2\phi(t_4)-2\phi(t_5)-2\phi(t_7)\\
    &\wedge&0&=-\phi(t_2)-\phi(t_3)-\phi(t_4)+\phi(t_5)\\
    &\wedge&0&=\phi(t_2)-\phi(t_3)\\
    &\wedge&0&=\phi(t_7)-\phi(t_9)\\
    &\wedge&0&=\phi(t_3)+\phi(t_4)-\phi(t_7)+\phi(t_9)\;,
  \end{align*}}\fi where the last four equations correspond
   to~$K_{G_\mathrm{ex}}^0$.
\end{example}

By using classical linear algebra results
\citep[e.g.,][Thm.~1]{pottier91}, in \appref{sec-trans} we prove the
following characterisation of the bounded variables of~$E_\xi$.
\begin{restatable}{lemma}{unbounded}\label{lem:unbounded}
  Assume that~$\xi=(\vec x_0G_0\vec y_0)\vec a_1\dots(\vec x_kG_k\vec
  y_k)$ is satisfiable.  Then for every $0\leq j\leq k$
  we have:
  \begin{itemize}
  \item For every $1\leq i\leq d$, the set of values
    $\vec{m}_j^{\vec{h}}(i)$ where $\vec{h}$ is a model of $E_\xi$ is unbounded if, and only if, there exists 
    a model $\vec{h}_0$ of $E_\xi^0$ such that $\vec{m}_j^{\vec{h}_0}(i)>0$.
  \item For every $t\in T_j$, the set of values
    $\phi_j^{\vec{h}}(t)$ where $\vec{h}$ is a model of $E_\xi$ is unbounded if, and only if, there exists 
    a model $\vec{h}_0$ of $E_\xi^0$ such that $\phi_j^{\vec{h}_0}(t)>0$.
  \item For every $1\leq i\leq d$, the set of values
    $\vec{n}_j^{\vec{h}}(i)$ where $\vec{h}$ is a model of $E_\xi$ is unbounded if, and only if, there exists 
    a model $\vec{h}_0$ of $E_\xi^0$ such that $\vec{n}_j^{\vec{h}_0}(i)>0$.
  \end{itemize}
  Moreover, the sum of the bounded values of~$E_\xi$ is bounded by
  $\size{\xi}^{\size{\xi}-1}$.
\end{restatable}

\section{The Decomposition Algorithm}
\label{sec-klm}
Let us give an overview of the decomposition algorithm.  Given an
instance $(G,\vec c_\mathit{in},\vec c_\mathit{out})$ of
the \nameref{reach} problem, the algorithm takes as input the KLM
sequence $\xi_0\eqdef(\vec c_\mathit{in}G\vec c_\mathit{out})$.  In an
initialisation phase, the algorithm computes a finite set
$\mathrm{clean}(\xi_0)$ of so-called \emph{clean} KLM sequences
(see \cref{lem:initialreduction}) such that
$L_{\xi_0}=\bigcup_{\xi'_0\in\mathrm{clean}(\xi_0)}L_{\xi'_0}$.  At
each step of the algorithm, given a clean KLM sequence~$\xi$,
\begin{itemize}
\item either~$\xi$ is \emph{normal}, which is a condition that ensures that the action
language~$L_\xi$ is non-empty (see \cref{lem:normalnonempty}),
\item or we
can perform a \emph{decomposition step} as per
\cref{lem:inductivereduction},
which produces a \emph{finite} (possibly empty)
set~$\mathrm{dec}(\xi)$ of clean KLM sequences such that
$\rank{\xi'}<_\mathit{lex}\rank{\xi}$ for all
$\xi'\in\mathrm{dec}(\xi)$ and
$L_\xi=\bigcup_{\xi'\in\mathrm{dec}(\xi)}L_{\xi'}$.
\end{itemize}
Both the initialisation and the decomposition steps are the results of
elementary steps presented in \cref{sub-elem} and aiming to enforce
various properties on KLM sequences.

By repeatedly applying decomposition steps, the decomposition
algorithm explores a \emph{decomposition forest} labelled with clean
KLM sequences, where the roots are labelled by the
elements~$\xi'_0\in\mathrm{clean}(\xi_0)$, and where each node
labelled by a non-normal KLM sequence~$\xi$ has a child
labelled~$\xi'$ for each $\xi'\in\mathrm{dec}(\xi)$.  A
decomposition forest has finitely many roots, finite branching degree,
and, because the ranks decrease strictly along the branches and
$(\+N^{d},{<_\mathit{lex}})$ is well-founded, it has finite branches.
A decomposition forest is thus finite by K\H{o}nig's Lemma, and the
algorithm terminates.

Note that, in order to answer the \nameref{reach} problem, we only
need to explore a decomposition forest nondeterministically in
search of a leaf labelled by a normal KLM sequence.  However, a
\emph{full decomposition} $\mathrm{fdec}(\xi_0)$, which we define as
the set of all the normal KLM sequences in a decomposition forest
for~$\xi_0$, is computable, and such that
\begin{equation}\label{eq-fdec}
  L_{\xi_0}=\bigcup_{\xi'\in\mathrm{fdec}(\xi_0)}L_{\xi'}\;.
\end{equation}

\begin{remark}
  Note that decomposition steps are not deterministic, meaning that
  there might be several choices of sets~$\mathrm{dec}(\xi)$ for
  each~$\xi$.  Thus there might be several decomposition forests for a
  KLM sequence~$\xi_0$.  This does not impact the correctness of the
  algorithm; in fact, we know from~\citep{leroux15} that all the full
  decompositions one can obtain actually denote the same canonical
  \emph{ideal decomposition}.
\end{remark}

\subsection{Elementary Decomposition Steps}\label{sub-elem}

As will be further explained in \cref{sub-normal}, clean KLM sequences
are obtained in
\cref{lem:initialreduction} by performing a decomposition
into strongly connected components (\cref{sub-scc}), followed by
a \emph{saturation} step (\cref{sub-saturated}), and keeping only the
satisfiable KLM sequences according to their characteristic systems,
which were defined in \cref{sub-chara}.  A decomposition step
according to \cref{lem:inductivereduction} first
unfolds \emph{unpumpable} (\cref{sub-pump}) or \emph{bounded}
(\cref{sub-unbounded}) sequences, and then cleans up the resulting
sequences thanks to \cref{lem:initialreduction}.

\subsubsection{Strongly Connected KLM Sequences}\label{sub-scc}
A KLM sequence
$\xi=(\vec{x}_0G_0\vec{y}_0)\vec{a}_1\ifams\ldots$ $\else (\vec{x}_1G_1\vec{y}_1)\ldots\fi
\vec{a}_k(\vec{x}_kG_k\vec{y}_k)$
is said to be \emph{strongly connected} if the VASSes $G_0,\ldots,G_k$
occurring in~$\xi$ are strongly connected.
\begin{lemma}\label{sec:sc}
  For any KLM sequence~$\xi$ that is not strongly connected, we can
  compute in time $\exp(\size{\xi})$ a finite set~$\Xi$ of
  strongly connected KLM sequences such that
  $L_\xi=\bigcup_{\xi'\in\Xi}L_{\xi'}$ and such that
  $\rank{\xi'}<_\mathit{lex}\rank{\xi}$ and
  $\size{\xi'}\leq \size{\xi}$ for every~$\xi'\in\Xi$.
\end{lemma}
\begin{proof}
  We just replace every triple $(\vec{x}G\vec{y})$ occurring in $\xi$
  where $G=(Q,q_\mathit{in},q_\mathit{out},T)$ is a non strongly
  connected VASS by all the possible sequences
  $(\vec{x}G_0\vec\omega)\vec{a}_1\ldots (\vec\omega G_n\vec{y})$
  where $n\geq 1$, $G_j=(Q_j,r_j,s_j,T_j)$ is such that
  $Q_0,\ldots,Q_n$ are distinct strongly connected components of $G$,
  $T_j\eqdef T\cap (Q_j\times\setZ^d\times Q_j)$ for every
  $0\leq j\leq n$, $r_0\eqdef q_\mathit{in}$,
  $s_n\eqdef q_\mathit{out}$, and $(s_{j-1},\vec{a}_j,r_j)$ is a
  transition in~$T$ for every $1\leq j\leq n$.\ifams\par\fi  We obtain
  that way a finite set~$\Xi$ of strongly connected KLM sequences
  satisfying the lemma.
  \ifams
  In particular, regarding sizes, observe that $|(\vec
  xG_0\vec\omega)\vec a_1\dots(\vec\omega G_n\vec
  y)|=2(d+1)^{d+1}(\|\vec x\|+\|\vec y\|+(n+\sum_{j=1}^n\|\vec
  a_j\|+\sum_{j=0}^n|G_j|))\leq 2(d+1)^{d+1}(\|\vec x\|+\|\vec y\|+|G|)$.\fi  
\end{proof}

\begin{figure}[tbp]
  \centering
  \begin{tikzpicture}[auto,on grid,node distance=1.6,initial text={}]
    \node[state,initial by arrow,accepting by arrow](p){$q_\mathit{in}$};
    \node[state,below left=of p](p1){$p$};
    \node[state,right=3 of p,initial by arrow,accepting by arrow](q){$q_\mathit{out}$};
    \node[state,below right=of q](q1){$q$};
    \path[->,every node/.style={font=\scriptsize,inner sep=1pt}]
      (p)  edge[loop above] node {$\vec{a}_1{=}(0,2,0)$}  ()
      (p)  edge[bend left,inner sep=0pt,near end]  node {$\vec{a}_2{=}(2,2,-1)$} (p1)
      (p1) edge[bend left, near start]  node {$\vec{a}_5{=}(1,0,-2)$} (p)
      (q)  edge[loop above] node {$\vec{a}_6{=}(1,-1,0)$} ()
      (q)  edge[bend left, near end]  node {$\vec{a}_7{=}(1,-1,-2)$} (q1)
      (q1) edge[bend left,inner sep=0pt,near start]  node {$\vec{a}_9{=}(0,0,0)$}  (q)
      (q1) edge[loop below,inner sep=3pt] node {$\vec{a}_8{=}(-2,-1,0)$}();
  \end{tikzpicture}
  \caption{\label{fig-scc}The strongly connected VASSes $G_\mathrm{ex}^1$ (left) and
  $G_\mathrm{ex}^2$ (right).}
\end{figure}
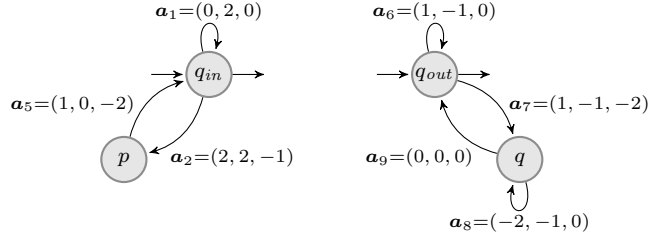
\begin{example}\label{ex-scc}
  Consider again the VASS $G_\mathrm{ex}$ of~\cref{ex-vass} and the
  KLM sequence $\xi_\mathrm{ex}=((0,0,2)G_\mathrm{ex}(1,1,0))$.  The
  decomposition into strongly connected KLM sequences yields
  a set $\{\xi_\mathrm{ex}^1,\xi_\mathrm{ex}^2\}$ where
  \begin{align*}
   \xi_\mathrm{ex}^1&\eqdef
  ((0,0,2)G_\mathrm{ex}^1(\omega,\omega,\omega))\vec a_3((\omega,\omega,\omega)G_\mathrm{ex}^2(1,1,0))\;,\\
   \xi_\mathrm{ex}^2&\eqdef
  ((0,0,2)G_\mathrm{ex}^1(\omega,\omega,\omega))\vec
  a_4((\omega,\omega,\omega)G_\mathrm{ex}^2(1,1,0))\;,
  \end{align*}
  where $G_\mathrm{ex}^1$ and~$G_\mathrm{ex}^2$ are displayed
  in \cref{fig-scc}.
\end{example}

\subsubsection{Saturated KLM Sequences}\label{sub-saturated}
A KLM sequence $\xi=(\vec{x}_0G_0\vec{y}_0)\vec{a}_1\ifams\relax\else (\vec{x}_1G_1\vec{y}_1)\fi\ldots
\vec{a}_k(\vec{x}_kG_k\vec{y}_k)$ is said to be \emph{saturated} if
for every $0\leq j\leq k$ and for every $i\in\{1,\ldots,d\}$ the following
two conditions hold:
\begin{enumerate}
\item if $\vec{x}_j(i)=\omega$, then the set of values $\vec{m}_j^{\vec{h}}(i)$
  where
  $\vec{h}$
  is a model of $E_\xi$ is unbounded, and
\item  if $\vec{y}_j(i)=\omega$, then the set of values $\vec{n}_j^{\vec{h}}(i)$
  where
  $\vec{h}$
  is a model of $E_\xi$ is unbounded.
\end{enumerate}
Saturation corresponds essentially to \citeauthor{kosaraju82}'s
property~$\uptheta1(b)$.

\begin{lemma}\label{lem:saturated}
  From any strongly connected KLM sequence $\xi$, we can compute in
  time $\exp(\size{\xi}^{\size{\xi}})$
  a
  finite set $\Xi$ of saturated strongly connected KLM sequences such that
  $L_\xi=\bigcup_{\xi'\in \Xi}L_{\xi'}$, and such that
  $\rank{\xi'}\leq_\mathit{lex}\rank{\xi}$ and $\size{\xi'}\leq
  \size{\xi}^{\size{\xi}}$ for every $\xi'\in\Xi$.
\end{lemma}
\begin{proof}
  Thanks to \cref{lem:unbounded}, we can saturate a KLM sequence.  In
  fact, we just have to replace some~$\omega$ components by all the
  possible bounded values $\leq\size{\xi}^{\size\xi-1}$ given by the
  characteristic system~$E_\xi$ for the variables~$\vec m_j,\vec n_j$.
\end{proof}

\begin{example}\label{ex-saturated}
  Consider the KLM sequences $\xi_\mathrm{ex}^1$ and
  $\xi_\mathrm{ex}^2$ from \cref{ex-scc}.  \Cref{lem:saturated} yields
  respectively
  \begin{align*}
  \xi_\mathrm{ex}^3&\eqdef
  ((0,0,2)G_\mathrm{ex}^1(0,\omega,2))\vec
  a_3((1,\omega,2)G_\mathrm{ex}^2(1,1,0))\;,\\ \xi_\mathrm{ex}^4&\eqdef
  ((0,0,2)G_\mathrm{ex}^1(0,\omega,2))\vec
  a_4((0,\omega,0)G_\mathrm{ex}^2(1,1,0))\;.\;\:\qedhere \end{align*}
\end{example}

\subsubsection{Unbounded KLM Sequences}\label{sub-unbounded}
Consider a KLM sequence~$\xi$ of form
$(\vec{x}_0G_0\vec{y}_0)\vec{a}_1\ifams$ $\else(\vec{x}_1G_1\vec{y}_1)\fi\ldots
\vec{a}_k (\vec{x}_kG_k\vec{y}_k)$,
where~$T_j$ denotes the set of transitions of~$G_j$.  Observe that, if
a transition~$t$ in~$T_j$ is such that the set of
values~$\phi_j^{\vec{h}}(t)$ where~$\vec{h}$ ranges over the models of
the characteristic system~$E_\xi$ of~$\xi$ is bounded by some
value~$B$, then the number of times a word $\sigma\in L_\xi$ can use
the transition~$t$ is bounded by~$B$.  It means that the VASS~$G_j$
can be replaced by at most~$B$ copies of itself without the
transition~$t$, joined using the action~$\Delta(t)$ of~$t$, while
preserving the language~$L_\xi$.

Formally, in such a situation, we define~$T'_j$ as the set of
transitions $t\in T_j$ such that the set of
values~$\phi^{\vec{h}}_j(t)$ is unbounded.  A KLM sequence~$\xi$ is
said to be \emph{unbounded} if~$T_j'=T_j$ for every $0\leq j\leq k$,
and otherwise to be \emph{bounded}.  Unboundedness corresponds
essentially to \citeauthor{kosaraju82}'s property~$\uptheta1(a)$, but
here we also need to show that the ranks decrease when performing this
decomposition.

\begin{lemma}\label{lem:unboundeddec}
  Whether a KLM sequence~$\xi$ is unbounded is in \NP.  Moreover, if
  $\xi$ is strongly connected and bounded, we can compute in time
  $\exp(\size{\xi}^{\size{\xi}})$ a finite set~$\Xi$ of KLM sequences
  such that $L_\xi=\bigcup_{\xi'\in\Xi}L_{\xi'}$ and such that
  $\rank{\xi'}<_\mathit{lex}\rank{\xi}$ and
  $\size{\xi'}\leq \size{\xi}^{\size{\xi}}$ for every~$\xi'\in\Xi$.
\end{lemma}
\begin{proof}
  Let $T_j'$ be the set of transitions $t\in T_j$ such that the set
  $\phi_j^{\vec{h}}(t)$ where $\vec{h}$ is a model of $E_\xi$ is unbounded.
  Let us introduce the VASS~$G_j'$ obtained from~$G_j$ by
  replacing~$T_j$ by~$T_j'$. Let $\vec{V}_j$ be
  the vector space spanned by the displacements of the cycles
  of~$G_j$, and let $\vec{V}_j'$ be the vector space generated by the
  displacements of the cycles of~$G_j'$. Since $T_j'\subseteq T_j$,
  naturally $\vec{V}_j'\subseteq \vec{V}_j$. We are going to prove
  that if $\vec{V}_j'=\vec{V}_j$ then $T_j'=T_j$.

  \begin{claim}\label{lem:redbound}
    Assume that $E_\xi$ is satisfiable. For every $j$, if
    $\vec{V}_j'=\vec{V}_j$ then $T_j'=T_j$.
  \end{claim}
  \begin{proof}[Proof of \cref{lem:redbound}]
    Let us consider $j\in\{0,\ldots,k\}$ such that
    $\vec{V}_j'=\vec{V}_j$ and let us prove that $T_j'=T_j$.
    By summing up a finite number of solutions of~$E_\xi^0$ (one for
    each transition~$t\in T_j'$), \cref{lem:unbounded} shows that
    there exists a solution $\vec{h}_0$ of~$E_\xi^0$ such that
    $\phi_j^{\vec{h}_0}(t)>0$ for every~$t\in T_j'$.

    Let us consider a cycle of~$G_j$ that contains all the transitions
    of~$T_j$; such a cycle exists since~$G_j$ is strongly connected.
    We denote by~$\psi$ the Parikh image of that cycle.  Notice that
    $\Delta(\psi)\in\vec{V}_j$; since $\vec{V}_j=\vec{V}_j'$, there
    exists a sequence $\theta_1,\ldots,\theta_s$ of cycles of~$G_j'$,
    and a sequence $\lambda_1,\ldots,\lambda_s$ of rational numbers
    such that $\Delta(\psi)=\sum_{r=1}^s\lambda_r\Delta(\phi_r)$,
    where $\phi_r$ is the Parikh image of $\theta_r$.  Let~$p>0$ be a
    natural number such that $p \lambda_r\in\setZ$ for every~$r$.
    Since $\phi_j^{\vec{h}_0}(t)>0$ for every $t\in T_j'$, there
    exists $q\in\setN$ such that
    $p \lambda_r\phi_r\leq q\phi_j^{\vec{h}_0}$ for every~$r$.  It
    follows that $\phi_r'\eqdef q\phi_j^{\vec{h}_0}-p \lambda_r\phi_r$
    maps every $t\in T_j\setminus T_j'$ to zero. Let~$\phi'$ be the mapping
    $p\psi+\sum_{r=1}^s\phi_r'$.  We deduce that
    \begin{equation}\Delta(\phi')=\Delta(qs\phi_j^{\vec{h}_0})=qs\vec{n}_j^{\vec{h}_0}-qs\vec{m}_j^{\vec{h}_0}\end{equation}
    since~$\vec{h}_0$ is a model of~$E_\xi^0$.

    It follows that the sequence~$\vec{h}_0'$ obtained from
    $qs \vec{h}_0$ by replacing the~$j$th tuple by
    $(qs\vec{m}_j^{\vec{h}_0},\phi',qs\vec{n}_j^{\vec{h}_0})$ is a
    model of~$E_\xi^0$. Notice that
    $\phi_j^{\vec{h}_0'}(t)=\phi'(t)\geq p\psi(t)\geq 1$ for every
    $t\in T_j$. \cref{lem:unbounded} shows that $T_j\subseteq
    T_j'$. Hence~$T_j'=T_j$.
  \end{proof}

  Let us return to the proof of~\cref{lem:unboundeddec}.  First
  observe that we can decide in nondeterministic polynomial time
  whether~$E_\xi$ is satisfiable. If it is not the case, then~$L_\xi$
  is empty and we can return the empty set.
  Otherwise, \cref{lem:unbounded} shows that the sets
  $T_1',\ldots,T_j'$ are computable in polynomial time. If $T_j'=T_j$
  for every~$j$, then~$\xi$ is unbounded. Otherwise,~$\xi$ is bounded,
  and there exists~$j$ such that~$T_j'$ is strictly included
  in~$T_j$. \Cref{lem:unbounded} shows that a word $\sigma\in L_\xi$
  cannot use a transition in $T_j\backslash T_j'$ more than
  $\size{\xi}^{\size{\xi}-1}$ times. It follows that we can replace
  the triple $\vec{x}_jG_j\vec{y}_j$ in~$\xi$ by a triple where the
  transitions in $T_j\setminus T'_j$ are taken at most
  $\size{\xi}^{\size{\xi}-1}$ times the VASS~$G_j'$.  Hence
  $\size{\xi'}\leq \size{\xi}^{\size{\xi}}$.  Since~$G_j$ is strongly
  connected, \cref{lem:redbound} shows that the KLM sequences~$\xi'$
  obtained that way satisfy $\rank{\xi'}<_\mathit{lex}\rank{\xi}$.
\end{proof}

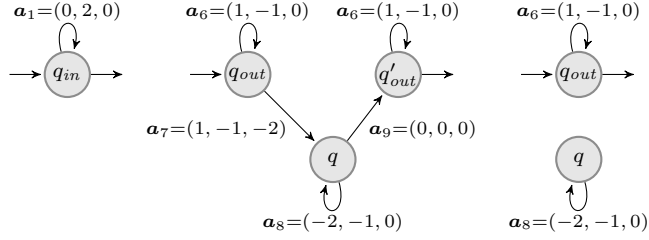
\begin{figure}[tbp]
  \centering\hspace*{-2ex}
  \begin{tikzpicture}[auto,on grid,node distance=1.6,initial text={}]
    \node[state,initial by arrow,accepting by arrow](p){$q_\mathit{in}$};
    \node[state,right=2.4 of p,initial by arrow](q){$q_\mathit{out}$};
    \node[state,below right=of q](q1){$q$};
    \node[state,right=2 of q,accepting by arrow](q2){$q'_\mathit{out}$};
    \node[state,right=2.4 of q2,initial by arrow,accepting by arrow](q3){$q_\mathit{out}$};
    \node[state,below=1.13 of q3](q4){$q$};
    
    \path[->,every node/.style={font=\scriptsize,inner sep=1pt}]
      (p)  edge[loop above] node {$\vec{a}_1{=}(0,2,0)$}  ()
      (q)  edge[loop above] node {$\vec{a}_6{=}(1,-1,0)$} ()
      (q)  edge[swap] node {$\vec{a}_7{=}(1,-1,-2)$} (q1)
      (q1) edge[loop below] node {$\vec{a}_8{=}(-2,-1,0)$}()
      (q1) edge[swap]  node {$\vec{a}_9{=}(0,0,0)$}  (q2)
      (q2)  edge[loop above] node {$\vec{a}_6{=}(1,-1,0)$} ()
      (q3)  edge[loop above] node {$\vec{a}_6{=}(1,-1,0)$} ()
      (q4) edge[loop below] node {$\vec{a}_8{=}(-2,-1,0)$}();
  \end{tikzpicture}
  \caption{\label{fig-bound}The VASSes $G_\mathrm{ex}^3$
      (left), $G_\mathrm{ex}^4$ (middle), and
  $G_\mathrm{ex}^5$ (right).}
  \end{figure}
\begin{example}\label{ex-bound}
  Consider the KLM sequences $\xi_\mathrm{ex}^3$ and
  $\xi_\mathrm{ex}^4$
  from \cref{ex-saturated}.  \Cref{lem:unboundeddec} yields
  respectively
  \begin{align*}
  \xi_\mathrm{ex}^5&\eqdef
  ((0,0,2)G_\mathrm{ex}^3(0,\omega,2))\vec
  a_3((1,\omega,2)G_\mathrm{ex}^4(1,1,0))\;,\\
  \xi_\mathrm{ex}^6&\eqdef
  ((0,0,2)G_\mathrm{ex}^3(0,\omega,2))\vec
  a_4((0,\omega,0)G_\mathrm{ex}^5(1,1,0))\;,
  \intertext{%
  where $G_\mathrm{ex}^3$, $G_\mathrm{ex}^4$, and $G_\mathrm{ex}^5$
  are displayed in \cref{fig-bound}.  Applying \cref{sec:sc}
  and \cref{lem:saturated} to $\xi_\mathrm{ex}^5$ yields}
  \xi_\mathrm{ex}^7&\eqdef
  ((0,0,2)G_\mathrm{ex}^3(0,\omega,2))\vec
  a_3((1,\omega,2)G_\mathrm{ex}^6(\omega,\omega,2))\vec
  a_7\\&\quad((\omega,\omega,0)G_\mathrm{ex}^7(\omega,\omega,0))\vec a_9((\omega,\omega,0)G_\mathrm{ex}^6(1,1,0))\;,
  \end{align*}
  where $G_\mathrm{ex}^6$ and $G_\mathrm{ex}^7$ are shown in \cref{fig-clean}.
  The KLM
  sequence~$\xi_\mathrm{ex}^6$ is unsatisfiable, thus
  by \cref{lem-unsat}, it can be discarded.  
\end{example}
\begin{figure}[tbp]
  \centering\hspace*{-2ex}
  \begin{tikzpicture}[auto,on grid,node distance=1.6,initial text={}]
    \node[state,initial by arrow,accepting by arrow](p){$q_\mathit{out}$};
    \node[state,right=4 of p,initial by arrow,accepting by arrow](q){$q$};    
    \path[->,every node/.style={font=\scriptsize,inner sep=1pt}]
      (p)  edge[loop above] node {$\vec{a}_6{=}(1,-1,0)$}  ()
      (q)  edge[loop above] node {$\vec{a}_8{=}(-2,-1,0)$} ();
  \end{tikzpicture}
  \caption{\label{fig-clean}The VASSes $G_\mathrm{ex}^6$
      (left) and $G_\mathrm{ex}^7$ (right).}
  \end{figure}
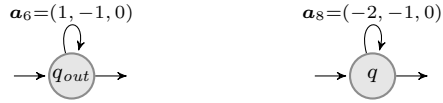

\subsection{Rigid KLM Sequences}\label{sub-rigid}
A component~$i$ is said to be \emph{fixed} by a VASS
$G=(Q,q_\mathit{in},q_\mathit{out},T)$ if there exists a function
$f_i{:}\,Q\rightarrow\setN$ such that $f_i(q)=f_i(p)+\vec{a}(i)$ for every
transition $(p,\vec{a},q)\in T$. Notice that we can compute in
polynomial time the set of fixed components of $G$, and given such a
component $i$, we can compute in polynomial time a function 
$f_i{:}\,Q\rightarrow\setN$ such that $f_i(q)=f_i(p)+\vec{a}(i)$ for every
transition $(p,\vec{a},q)\in T$.

A KLM sequence $\vec{x}G\vec{y}$ where
$G=(Q,q_\mathit{in},q_\mathit{out},T)$ is a VASS is said to be
\emph{rigid} if for every component~$i$ that is fixed by~$G$ there exists a function
$g_i{:}\,Q\rightarrow\setN$ such that $g_i(q)=g_i(p)+\vec{a}(i)$ for
every transition $(p,\vec{a},q)\in T$, and such that
$g_i(q_\mathit{in})\sqsubseteq \vec{x}(i)$ and
$g_i(q_\mathit{out})\sqsubseteq\vec{y}(i)$.  More generally, a KLM
sequence $\xi=(\vec{x}_0G_0\vec{y}_0)\vec{a}_1\ifams\relax\else
(\vec{x}_1G_1\vec{y}_1)\fi\ldots
\vec{a}_k(\vec{x}_kG_k\vec{y}_k)$ is said to be \emph{rigid} if
$\vec{x}_jG_j\vec{y}_j$ is rigid for every $0\leq j\leq k$.
Rigidity corresponds essentially to the rigid components introduced
by \citeauthor{kosaraju82}.  \ifams\relax\else We prove the following
in \appref{sec-rigid}.\fi

\begin{restatable}{lemma}{rigid}\label{lem:rigid}
  From any strongly connected KLM sequence~$\xi$, we can decide in
  time $\poly(\size{\xi})$ whether~$\xi$ is not rigid.  Moreover, in
  that case we can compute in time $\poly(\size{\xi})$ a KLM
  sequence~$\xi'$ such that $L_\xi=L_{\xi'}$,
  $\rank{\xi'}<_\mathit{lex}\rank{\xi}$, and
  $\size{\xi'}\leq \size{\xi}$.
\end{restatable}
\ifams\begin{proof}
  Let us assume that~$\xi$ is the KLM sequence $\vec{x}G\vec{y}$ where
  $G=(Q,q_\mathit{in},q_\mathit{out},T)$ is strongly connected (the
  general case can be obtained the same way).  We can compute in
  polynomial time by a straightforward constant propagation algorithm
  the set~$I$ of components that are fixed by~$G$ and for every~$i\in
  I$ a function $f_i{:}\,Q\rightarrow\setN$ such that
  $f_i(q)=f_i(p)+\vec{a}(i)$ for every transition $(p,\vec{a},q)\in
  T$.
    
  \begin{claim}\label{cl-rigid}
    $\xi$ is rigid if and only if the following three conditions hold
    for every~$i\in I$ and for every~$q\in Q$:
  \begin{enumerate}[(i)]
  \item\label{rigid-1} $\vec{y}(i)-f_i(q_\mathit{out})=\vec{x}(i)-f_i(q_\mathit{in})$ if
    $\vec{x}(i),\vec{y}(i)\in\setN$,
  \item\label{rigid-2} $\vec{x}(i)-f_i(q_\mathit{in})+f_i(q)\geq 0$ if
    $\vec{x}(i)\in\setN$, and
  \item\label{rigid-3} $\vec{y}(i)-f_i(q_\mathit{out})+f_i(q)\geq 0$ if
    $\vec{y}(i)\in\setN$.
  \end{enumerate}
  \end{claim}
  \begin{proof}[Proof of \cref{cl-rigid}]
    Assume first that~$\xi$ is rigid.  In that case, for every $i\in
    I$ there exists a function $g_i{:}\,Q\rightarrow\setN$ such that
    $g_i(q)=g_i(p)+\vec{a}(i)$ for every transition $(p,\vec{a},q)\in
    T$ and such that $g_i(q_\mathit{in})\sqsubseteq \vec{x}(i)$ and
    $g_i(q_\mathit{out})\sqsubseteq\vec{y}(i)$.  Since $G$ is strongly
    connected, it follows that there exists an integer $z_i\in\setZ$
    such that $g_i(q)=z_i+f_i(q)$ for every $q\in Q$. This equality in
    $q_\mathit{in}$ and $q_{\mathit{out}}$ provides
    $z_i=\vec{x}(i)-f_i(q_\mathit{in})$ if $\vec{x}(i)\in\setN$ and
    $z_i=\vec{y}(i)-f_i(q_\mathit{out})$ if $\vec{y}(i)\in\setN$. We
    deduce that conditions \eqref{rigid-1}, \eqref{rigid-2},
    and \eqref{rigid-3} hold.

    Conversely, assume that these conditions hold and let us prove
    that~$\xi$ is rigid. Let $i\in I$ and let us prove that there
    exists a function $g_i{:}\,Q\rightarrow\setN$ such that
    $g_i(q)=g_i(p)+\vec{a}(i)$ for every transition $(p,\vec{a},q)\in
    T$ and such that $g_i(q_\mathit{in})\sqsubseteq \vec{x}(i)$ and
    $g_i(q_\mathit{out})\sqsubseteq\vec{y}(i)$.  If
    $\vec{x}(i)=\omega$ and $\vec{y}(i)=\omega$, notice that
    $g_i\eqdef f_i$ fullfills the required conditions.  If
    $\vec{x}(i)\in\setN$, then
    condition~\eqref{rigid-1} shows that we define
    $g_i{:}\,Q\rightarrow\setN$ by
    $g_i(q)\eqdef\vec{x}(i)-f_i(q_\mathit{in})+f_i(q)$.  Notice that
    for every transition $(p,\vec{a},q)\in T$, we have
    $g_i(q)=g_i(p)+\vec{a}(i)$.  Observe that
    $g_i(q_\mathit{in})= \vec{x}(i)$.  Let us show that
    $g_i(q_\mathit{out})\sqsubseteq\vec{y}(i)$. If
    $\vec{y}(i)=\omega$, the relation is immediate. Otherwise, by
    condition~\eqref{rigid-1}, we get
    $g_i(q_\mathit{out})= \vec{y}(i)$. We have proved that $g_i$
    fullfills the required conditions. Symmetrically, we
    obtain the case $\vec{y}(i)\in\setN$ and
    $\vec{x}(i)\in\setN_\omega$.  
    We have shown that~$\xi$ is
    rigid.
  \end{proof}
  
  By \cref{cl-rigid}, we can decide in polynomial time whether~$\xi$
  is rigid.  Moreover, if~$\xi$ is not rigid, we can compute in
  polynomial time both~$i\in I$ and~$q\in Q$ such that one of the
  three conditions \eqref{rigid-1}, \eqref{rigid-2},
  and \eqref{rigid-3} does not hold.  If condition~\eqref{rigid-1}
  does not hold, then~$\xi$ cannot be satisfiable, and in particular
  $L_\xi=\emptyset$.  Thus we can consider for~$\xi'$ the KLM sequence
  obtained from~$\xi$ by removing all the transitions and all the
  states except~$q_\mathit{in}$ and~$q_\mathit{out}$.

  Otherwise, if condition~\eqref{rigid-1} holds, then
  either~\eqref{rigid-2} or~\eqref{rigid-3} does not hold.
  Since~\eqref{rigid-1} holds, it follows that
  $q\not\in\{q_\mathit{in},q_\mathit{out}\}$.  Let us show that
  $L_\xi=L_{\xi'}$ where $\xi'\eqdef\vec{x}G'\vec{y}$ and
  $G'\eqdef(Q',q_\mathit{in},q_\mathit{out},T')$, $Q'\eqdef
  Q\backslash\{q\}$, $T'\eqdef T\cap (Q'\times\setZ^d\times Q')$. To
  prove this inclusion, let us consider any $\sigma\in L_\xi$. There
  exists two configurations $\vec{m}$ and $\vec{n}$ and a word
  $\sigma=\vec{a}_1\ldots\vec{a}_k$ of actions such that
  $q_\mathit{in}(\vec{m})\xrightarrow[G]{\sigma}q_\mathit{out}(\vec{n})$. Thus
  there exists a sequence $q_0(\vec{c}_0),\ldots,q_k(\vec{c}_k)$ of
  state-configurations such
  that \begin{equation}q_\mathit{in}(\vec{m})=q_0(\vec{c}_0)\xrightarrow[G]{\vec{a}_1}\cdots \xrightarrow[G]{\vec{a}_k}q_k(\vec{c}_k)=q_\mathit{out}(\vec{n})\;.\end{equation}
  Observe that if $\vec{x}(i)\in\setN$, then $\vec{c}_0(i)=\vec{x}(i)$
  and by induction we get
  $\vec{x}(i)-f_i(q_\mathit{in})+f_i(q_j)=\vec{c}_j(i)\geq 0$ for
  every $0\leq j\leq k$.  Symmetrically, if $\vec{y}(i)\in\setN$, then
  $\vec{y}(i)-f_i(q_\mathit{out})+f_i(q_j)=\vec{c}_j(i)\geq 0$ for
  every $0\leq j\leq k$.  Thus $q\not\in \{q_0,\ldots,q_k\}$ and in
  particular $\sigma\in L_{\xi'}$.
\end{proof}
\fi

\subsubsection{Pumpable KLM Sequences}\label{sub-pump}
\newcommand{\Facc}[2]{\operatorname{Facc}_{#1}(#2)}
\newcommand{\Bacc}[2]{\operatorname{Bacc}_{#1}(#2)}

Given a VASS $G$ and two configurations $\vec{x},\vec{y}$, 
the \emph{forward and backward accelerations} are the vectors
$\Facc{G}{\vec{x}}$ and $\Bacc{G}{\vec{y}}$ in~$\setN_\omega^d$
defined respectively for every
$i\in\{1,\ldots,d\}$ as follows:
\begin{align*}
  \Facc{G}{\vec{x}}(i)&\eqdef\begin{cases}
    \omega & \text{ if }\exists \vec{x}'\geq\vec{x}\text{ with
    }\vec{x}'(i)>\vec{x}(i)\text{ s.t.}\\
    & q_\mathit{in}(\vec{x})\xrightarrow[G]{*}q_\mathit{in}(\vec{x}')\\
    \vec{x}(i) & \text{otherwise}
\end{cases}\\
  \Bacc{G}{\vec{y}}(i)&\eqdef\begin{cases}
    \omega & \text{ if }\exists \vec{y}'\geq\vec{y} \text{ with }
    \vec{y}'(i)>\vec{y}(i)\text{ s.t.}\\& q_\mathit{out}(\vec{y}')\xrightarrow[G]{*}q_\mathit{out}(\vec{y})\\
    \vec{y}(i) & \text{otherwise}
\end{cases}
\end{align*}

Observe that $\Facc{G}{\vec{x}}(i)=\vec{x}(i)$ and
$\Bacc{G}{\vec{y}}(i)=\vec{y}(i)$ for every component~$i$ fixed
by~$G$.  A triple $(\vec{x}G\vec{y})$ is said to be \emph{pumpable} if
$\Facc{G}{\vec{x}}(i)=\omega$ and $\Bacc{G}{\vec{y}}(i)=\omega$ for
every component~$i$ not fixed by~$G$.  More generally, a KLM sequence
$\xi=(\vec{x}_0G_0\vec{y}_0)\vec{a}_1 (\vec{x}_1G_1\vec{y}_1)\ldots
\vec{a}_k (\vec{x}_kG_k\vec{y}_k)$
is said to be \emph{pumpable} if $(\vec{x}_jG_j\vec{y}_j)$ is pumpable
for every $0\leq j\leq k$, and otherwise to be
\emph{unpumpable}.

\begin{remark}
  Pumpability, rigidity, and saturation together correspond
  essentially to \citeauthor{kosaraju82}'s property~$\uptheta2$.  In
  fact, we show in \appref{sub:flow} that if 
  a KLM sequence $\vec{x}G\vec{y}$ is pumpable, rigid, and
  saturated, then there exists a function
  $f{:}\,Q\rightarrow\setN_\omega^d$ such that $f(q)=f(p)+\vec{a}$ for
  every $(p,\vec{a},q)\in T$, and such that
  $f(q_\mathit{in})=\Facc{G}{\vec{x}}$ and
  $f(q_\mathit{out})=\Bacc{G}{\vec{y}}$. 
\end{remark}

\begin{example}\label{ex-unpump}
  The KLM sequence $\xi_\mathrm{ex}^7$ from \cref{ex-bound} is
  unpumpable: indeed, in the triple
  $((\omega,\omega,0)G^6_\mathrm{ex}(1,1,0))$, the components~$1$
  and~$2$ are not fixed\ifams, but we find \else\ but \fi
  $\Bacc{G_\mathrm{ex}^6}{(1,1,0)}(1)=\Bacc{G_\mathrm{ex}^6}{(1,1,0)}(2)=1$.
\end{example}

\paragraph{Deciding Pumpability}

Observe that $\Facc{G}{\vec{x}}$ and $\Bacc{G}{\vec{y}}$ are
computable by performing $2d$ calls to an oracle for the
\emph{coverability problem}~\citep[see, e.g.,][Lem.~3.3]{leroux10}.
By the results of \citet{rackoff78}, we can therefore decide in
exponential space whether a KLM sequence~$\xi$ is pumpable.

\paragraph{Unfolding}

When a KLM sequence~$\xi$ is unpumpable, there is a triple
$(\vec xG\vec y)$ and a component~$i$ not fixed by~$G$ such that
$\Facc{G}{\vec{x}}(i)<\omega$ or $\Bacc{G}{\vec{y}}(i)<\omega$.
Assume that we are in the former case.
If~$\xi$ is strongly connected, then there exists a finite $B\in\+N$
such that $\Facc{G}{\vec{x}}(i)=B$, and
the idea is then to
\emph{unfold}~$G$ by tracking the value of the $i$th component in the
control state.  Classically, such a bound~$B$ is computed by
constructing a \citeauthor{karp69} coverability tree, but this has a
worst-case Ackermannian complexity~\citep{cardoza76}.  Thus the
decomposition algorithms of \citeauthor{mayr81},
\citeauthor{kosaraju82}, and \citeauthor{lambert92} might use
an Ackermannian time in their very first decomposition step.

Here, we refine this decomposition step using insights from
\citeauthor{rackoff78}'s results in~\citep{rackoff78}.  We show that,
if there is a component~$i$ not fixed by~$G$ such that
$\Facc{G}{\vec{x}}(i)<\omega$, then there exists a component~$i'$ not
fixed by~$G$ and such that a double exponential~$B$ suffices.  Formally, let
$\setN_B\eqdef\{0,\ldots,B-1,\omega\}$.  Consider any
$i\in\{1,\ldots,d\}$, $r\in\setN_B$, and $\vec{x}(i)\in\setN_B$; the
\emph{forward $(i,B,r)$-unfolding} of a KLM triple $\vec{x}G\vec{y}$
is the KLM triple $\vec{x}G'\vec{y}$ where
$G'\eqdef(Q\times\setN_B,(q_\mathit{in},\vec{x}(i)),(q_\mathit{out},r),T')$
and~$T'$ is the set of transitions $((p,m),\vec{a},(q,n))$ where
$(p,\vec{a},q)\in T$ and $m,n\in\setN_B$ satisfy $n=m+\vec{a}(i)$ or
$(n=\omega\wedge m+\vec{a}(i)\geq B)$, and such that $m=\omega$
implies $q\neq q_\mathit{in}$.  (The \emph{backward
  $(i,B,r)$-unfolding} is defined symmetrically.)  We show the
following in \appref{sec-pump}.

\begin{restatable}{lemma}{cycleplus}\label{lem:cycleplus}
  Let $\xi=\vec{x}G\vec{y}$ be a KLM sequence and let~$I$ be the set
  of components $i\in\{1,\ldots,d\}$ that are not fixed by~$G$ and
  such that $\Facc{G}{\vec{x}}(i)<\omega$.  If~$I$ is not empty, then
  there exists $i\in I$ such that
  $L_{\xi}=\bigcup_{r\in\setN_B}L_{\xi_r}$ where~$\xi_r$ is the
  forward $(i,B,r)$-unfolding of~$\xi$ and
  $B\eqdef (\norm{\vec{x}}{}+2\size{G})^{1+d^d}$.
\end{restatable}

Of course, we also require that unfolding $\vec{x}G\vec{y}$ decreases
the rank.  The condition that $m=\omega$ must imply
$q\neq q_\mathit{in}$ in the unfolding is central for the proof of the
following lemma\ifams\relax\else, shown in \appref{sec-pump}\fi.

\begin{restatable}{lemma}{rankred}\label{lem:rankred}
  Let $\xi=\vec{x}G\vec{y}$ be a strongly connected KLM sequence and
  let~$i$ be a component not fixed by~$G$ and such that
  $\vec{x}(i)\in\setN_B$ for some~$B\in\setN$.  Then the
  $(i,B,r)$-unfolding~$\xi'$ of~$\xi$ satisfies
  $\rank{\xi'}<_\mathit{lex}\rank{\xi}$ for all $r\in\+N_B$.
\end{restatable}
\ifams\begin{proof}
  Assume that $G=(Q,q_\mathit{in},q_\mathit{out},T)$ and let
  $\xi'=\vec{x}G'\vec{y}$ be the $(i,B,r)$-unfolding of $G$ where
  $G'=(Q\times\setN_B,(q_\mathit{in},\vec{x}(i)),(q_\mathit{out},r),T')$.
  Let~$\vec{V}$ be the vector space generated by the displacements of the
  cycles of~$G$. As~$G$ is strongly connected, \cref{lem:cyclespan}
  shows that $\VS(t)=\vec{V}$ for every transition~$t$ in~$T$.

  Observe that since~$i$ is not fixed by~$G$, it means that there
  exists a vector~$\vec{v}\in\vec{V}$ such that $\vec{v}(i)\not=0$. In
  particular the dimension of~$\vec{V}$ is larger than or equal to~one.

  Let us observe that every cycle of~$G'$ labelled by a word~$\sigma$
  corresponds (by projecting on the first component of its control
  states) to a cycle of~$G$ also labelled by $\sigma$.  It follows
  that the displacement of every cycle of $G'$ is in $\vec{V}$,
  therefore $\vec{V}_{\!G'}(t')\subseteq \vec{V}$ for every
  transition~$t'$ in~$T'$.  Let us consider such a transition
  $t'=((p,m),\vec{a},(q,n))$ from~$T'$, such that $(p,\vec{a},q)\in T$
  and $m,n\in\setN_B$.  For the transitions $t'\in T'$ such that
  $m=\omega$, then $n=\omega$ and $q\neq q_\mathit{in}$, thus there
  are at most $|T|-1$ such transitions.  For the other transitions in
  $T'$, i.e., such that $m\not=\omega$, let us prove that $\vec
  V_{\!G'}(t')$ is strictly included in $\vec{V}$. If there is no
  cycle using~$t'$, then $\vec{V}_{G'}(t')=\{\vec{0}\}$ and we are
  done.  Otherwise, notice that this cycle keep tracks in~$G'$ of the
  precise displacement on the component~$i$ since there is no way to
  move from a state in $Q\times\{\omega\}$ to a state in
  $Q\times\{0,\ldots,B-1\}$.  It follows that the displacement of such
  a cycle is zero on component~$i$.  Hence the vector $\vec{v}$ we
  singled out earlier is not in~$\vec{V}_{\!G'}(t')$ and we have proven
  that $\vec{V}_{\!G'}(t')$ is strictly included in $\vec{V}$.

  This shows that
  $\rank{G'}<_\mathit{lex}\rank{G}$.
\end{proof}
\fi

Together, the previous two lemmas allow to show the following.
\begin{restatable}{lemma}{pump}\label{lem:pump}
  Whether a KLM sequence~$\xi$ is pumpable is in \EXPSPACE.  Moreover,
  if~$\xi$ is strongly connected and unpumpable, we can compute in
  time $\exp(\size{\xi}^{2+d^d})$ a finite set~$\Xi$ of KLM sequences
  such that $L_\xi=\bigcup_{\xi'\in\Xi}L_{\xi'}$ and such that
  $\rank{\xi'}<_\mathit{lex}\rank{\xi}$ and
  $\size{\xi'}\leq \size{\xi}^{2+d^d}$ for every $\xi'\in\Xi$.
\end{restatable}
\begin{proof}
  We have already argued that pumpability is decidable in exponential
  space.  Assume that~$\xi$ is strongly connected and unpumpable.
  Then there is a triple~$\vec{x}G\vec{y}$ in~$\xi$ and a
  component~$i$ not fixed by~$G$ such that
  $\Facc{G}{\vec{x}}(i)<\omega$ or $\Bacc{G}{\vec{y}}(i)<\omega$.
  Let us consider the former case and define
  $B\eqdef (\norm{\vec{x}}{}+2\size{G})^{1+d^d}$.

  \Cref{lem:cycleplus} shows that
  $L_\xi=\bigcup_{r\in\setN_B}L_{\xi_r}$ where $\xi_r$ is the KLM
  sequence obtained from~$\xi$ by replacing the KLM triple
  $\vec{x}G\vec{y}$ by its $(i',B,r)$-unfolding for a
  suitable~$i'$. \Cref{lem:rankred} shows that
  $\rank{\xi_r}<_\mathit{lex}\rank{\xi}$.  Finally,
  $B<\size{\xi}^{1+d^d}$ and thus
  $\size{\xi_r}\leq (1+B)\size{\xi}\leq \size{\xi}^{2+d^d}$.
\end{proof}

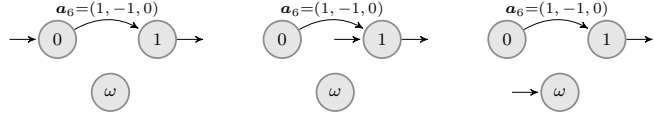
\begin{figure}[tbp]
  \centering\hspace*{-2ex}\scalebox{.83}{
  \begin{tikzpicture}[auto,on grid,node distance=1.2,initial text={}]
    \node[state,initial by arrow](p){$0$};
   \node[state,right=1.6 of p,accepting by arrow](q){$1$};
    \node[state,below right=of p]{$\omega$};
    \node[state,right=2 of q](p1){$0$}; 
   \node[state,right=1.6 of p1,initial by arrow,accepting by arrow](q1){$1$};   
    \node[state,below right=of p1]{$\omega$};
    \node[state,right=2 of q1](p2){$0$}; 
   \node[state,right=1.6 of p2,accepting by arrow](q2){$1$};   
    \node[state,below right=of p2,initial by arrow]{$\omega$};
    \path[->,every node/.style={font=\scriptsize,inner sep=1pt}]
      (p)  edge[bend left] node {$\vec{a}_6{=}(1,-1,0)$}  (q)
      (p1)  edge[bend left] node {$\vec{a}_6{=}(1,-1,0)$} (q1)
      (p2)  edge[bend left] node {$\vec{a}_6{=}(1,-1,0)$} (q2);
  \end{tikzpicture}}
  \caption{\label{fig-normal}The VASSes $G_\mathrm{ex}^8$
      (left), $G_\mathrm{ex}^9$ (middle), and~$G_\mathrm{ex}^{10}$ (right).}
  \end{figure}
\begin{example}\label{ex-pump}
  Consider again \cref{ex-unpump} and in particular component~$1$.
  Then $B=1$ suffices, and we can unfold along the first
  component, yielding three new KLM triples
  \begin{align*}
  \xi_\mathrm{ex}^8&\eqdef
  ((0,0,2)G_\mathrm{ex}^3(0,\omega,2))\vec
  a_3((1,\omega,2)G_\mathrm{ex}^6(\omega,\omega,2))\vec
  a_7\\&\quad((\omega,\omega,0)G_\mathrm{ex}^7(\omega,\omega,0))\vec a_9((\omega,\omega,0)G_\mathrm{ex}^8(1,1,0))\;,\\
  \xi_\mathrm{ex}^9&\eqdef
  ((0,0,2)G_\mathrm{ex}^3(0,\omega,2))\vec
  a_3((1,\omega,2)G_\mathrm{ex}^6(\omega,\omega,2))\vec
  a_7\\&\quad((\omega,\omega,0)G_\mathrm{ex}^7(\omega,\omega,0))\vec a_9((\omega,\omega,0)G_\mathrm{ex}^9(1,1,0))\;,\\
  \xi_\mathrm{ex}^{10}&\eqdef
  ((0,0,2)G_\mathrm{ex}^3(0,\omega,2))\vec
  a_3((1,\omega,2)G_\mathrm{ex}^6(\omega,\omega,2))\vec
  a_7\\&\quad((\omega,\omega,0)G_\mathrm{ex}^7(\omega,\omega,0))\vec a_9((\omega,\omega,0)G_\mathrm{ex}^{10}(1,1,0))\;,
  \end{align*}
  where~$G_\mathrm{ex}^8$, $G_\mathrm{ex}^9$, and~$G_\mathrm{ex}^{10}$ are shown
  in \cref{fig-normal}.  When applying \cref{sec:sc,lem:saturated},
  $\xi_\mathrm{ex}^8$ and~$\xi_\mathrm{ex}^9$ are respectively
  decomposed into
  \begin{align*}
  \xi_\mathrm{ex}^{11}&\eqdef
  ((0,0,2)G_\mathrm{ex}^3(0,\omega,2))\vec
  a_3((1,\omega,2)G_\mathrm{ex}^6(\omega,\omega,2))\vec
  a_7\\&\quad((\omega,\omega,0)G_\mathrm{ex}^7(0,2,0))\vec
         a_9((0,2,0)G_\mathrm{ex}^{11}(0,2,0))\vec a_6\\&\quad((1,1,0)G_\mathrm{ex}^{11}(1,1,0))\;,\\
  \xi_\mathrm{ex}^{12}&\eqdef
  ((0,0,2)G_\mathrm{ex}^3(0,\omega,2))\vec
  a_3((1,\omega,2)G_\mathrm{ex}^6(\omega,\omega,2))\vec
  a_7\\&\quad((\omega,\omega,0)G_\mathrm{ex}^7(1,1,0))\vec a_9((1,1,0)G_\mathrm{ex}^{11}(1,1,0))\;,
  \end{align*}
  where $G_\mathrm{ex}^{11}=(\{q\},\{q\},\{q\},\emptyset)$ is the
  trivial VASS with no transitions, while~$G_\mathrm{ex}^{10}$ is discarded.
\end{example}

\subsection{Normal KLM Sequences}\label{sub-normal}
A KLM sequence is said to be \emph{clean} if it is satisfiable
(see \cref{sub-chara}), strongly connected (see \cref{sub-scc}), and
saturated (see \cref{sub-saturated}).  It is \emph{normal} if it is
clean, rigid (see \cref{sub-rigid}), pumpable (see \cref{sub-pump}), and unbounded
(see \cref{sub-unbounded}).

\subsubsection{Cleaning Lemma}
We can transform any KLM sequence into a finite set of clean KLM
sequences thanks to the following lemma.
\begin{lemma}[Cleaning]\label{lem:initialreduction}
  From any KLM sequence $\xi$, we can compute in time
  $\exp(g(\size{\xi}))$ a finite set $\mathrm{clean}(\xi)$ of clean
  KLM sequences such that
  $L_\xi=\bigcup_{\xi'\in\mathrm{clean}(\xi)}L_{\xi'}$ and such that
  $\rank{\xi'}\leq_\mathit{lex}\rank{\xi}$ and
  $\size{\xi'}\leq g(\size\xi)$ for every
  $\xi'\in\mathrm{clean}(\xi)$, where $g(x)\eqdef x^{x}$.
\end{lemma}
\begin{proof}
  By \cref{sec:sc}, we can compute a finite set $\Xi$ of strongly
  connected KLM sequences such that
  $L_\xi=\bigcup_{\xi'\in\Xi}L_{\xi'}$ and such that
  $\rank{\xi'}\leq_\mathit{lex}\rank{\xi}$ and
  $\size{\xi'}\leq\size{\xi}$ for every $\xi'\in\Xi$. By
  applying \cref{lem:saturated} to each KLM sequence in $\Xi$, we
  compute in exponential time a finite set $\Xi'$ of saturated
  strongly connected KLM sequences such that
  $\bigcup_{\xi'\in\Xi}L_{\xi'}=\bigcup_{\xi''\in \Xi'}L_{\xi''}$ and
  such that $\rank{\xi''}\leq_\mathit{lex}\rank{\xi}$ and
  $\size{\xi''}\leq \size{\xi}^{\size{\xi}}$ for every $\xi''\in\Xi'$.
  By \cref{lem-unsat}, we can safely remove the unsatisfiable KLM
  sequences from $\Xi'$---which can be performed in nondeterministic
  time polynomial in $\sum_{\xi''\in\Xi''}\size{\xi''}$ since each
  $E_{\xi''}$ is of size polynomial in $\size{\xi''}$--- and we obtain
  a set $\mathrm{clean}(\xi)$ satisfying the lemma.
\end{proof}

\subsubsection{Decomposition Lemma}
In order to decompose a KLM sequence into a finite set of normal KLM
sequences, the decomposition algorithm applies as many times as
possible the \emph{decomposition step} defined by the following lemma.
\begin{lemma}[Decomposition]\label{lem:inductivereduction}
  Let $\xi$ be a clean KLM sequence.  If $\xi$ is not normal, we can
  compute in time $\exp(h(\size{\xi}))$ a finite set $\mathrm{dec}(\xi)$
  of clean KLM sequences such that
  $L_\xi=\bigcup_{\xi'\in\mathrm{dec}(\xi)}L_{\xi'}$ and such that
  $\rank{\xi'}<_\mathit{lex}\rank{\xi}$ and $\size{\xi'}\leq
  h(\size\xi)$ for every $\xi'\in\mathrm{dec}(\xi)$,
  where $h(x)\eqdef x^{x^{1+x}}$.
\end{lemma}
\begin{proof}
  \Cref{lem:rigid,lem:pump,lem:unboundeddec} show that we can compute in double
  exponential time a finite set $\Xi$ of KLM sequences such that
  $L_\xi=\bigcup_{\xi'\in\Xi}L_{\xi'}$ and such that
  $\rank{\xi'}<_\mathit{lex}\rank{\xi}$ and
  $\size{\xi'}\leq \size{\xi}^{\size{\xi}}$ for every
  $\xi'\in\Xi$ by observing that $2+d^d\leq \size{\xi}$.
  For each KLM sequence $\xi'\in \Xi$, by applying
  \cref{lem:initialreduction} we compute in time exponential in
  $g(\size{\xi'})$ a finite set $\mathrm{clean}(\xi')$ of clean KLM
  sequences such that
  $L_{\xi'}=\bigcup_{\xi''\in \mathrm{clean}(\xi')}L_{\xi''}$ and such
  that $\rank{\xi''}\leq_\mathit{lex}\rank{\xi'}$ and
  $\size{\xi''}\leq g(\size{\xi'})$ for each
  $\xi''\in\mathrm{clean}(\xi')$.  We deduce the statement by letting
  $\mathrm{dec}(\xi)\eqdef\bigcup_{\xi'\in \Xi}\mathrm{clean}(\xi')$.
\end{proof}

\subsubsection{Bounded Witness Lemma}\label{sub-normalnonempty}
Thanks to the following lemma, we can stop the decomposition once we
obtain a normal KLM sequence.  The proof given in \appref{sec-normal}
follows the same lines as \citeauthor{kosaraju82}'s, with the added
twist that we extract a bound on the length of minimal words
in~$L_\xi$.
\begin{restatable}[Bounded Witness]{lemma}{normalnonempty}\label{lem:normalnonempty}
  From any normal KLM sequence $\xi$, we can compute in space $O(\ell(\size{\xi}))$
  a word $\sigma\in
  L_\xi$ such that $|\sigma|\leq\ell(\size{\xi})$ where $\ell(x)\mathbin{\raisebox{0pt}[0pt][0pt]{$\eqdef$}} x^{3x}$.
\end{restatable}

\begin{figure}[tbp]
  \centering
  \begin{tikzpicture}[auto,on grid,node distance=1.2]
    \node(3){$\xi_\mathrm{ex}^3$};
    \node[right=3 of 3](4){$\xi_\mathrm{ex}^4$};
    \node[below=1 of 3](7){$\xi_\mathrm{ex}^7$};
    \node[below left=of 7](8){$\xi_\mathrm{ex}^{11}$};
    \node[below right=of 7](9){$\xi_\mathrm{ex}^{12}$};
    \path[-]
      (3) edge (7)
      (7) edge (8)
      (7) edge (9);
   \end{tikzpicture}
   \caption{\label{fig-forest}A decomposition forest for
  $\xi_\mathrm{ex}$.} 
\end{figure}
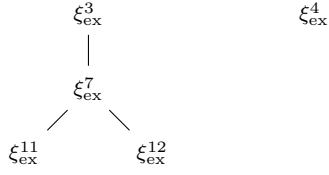
\begin{example}\label{ex-forest}
  Let us consider \crefrange{ex-scc}{ex-pump}.  We have
  $\mathrm{clean}(\xi_\mathrm{ex})=\{\xi_\mathrm{ex}^3,\xi_\mathrm{ex}^4\}$,
  which are both bounded, and then
  $\mathrm{dec}(\xi^3_\mathrm{ex})=\{\xi_\mathrm{ex}^7\}$ and
  $\mathrm{dec}(\xi^4_\mathrm{ex})=\emptyset$.  Then,
  $\xi_\mathrm{ex}^7$ is unpumpable and
  $\mathrm{dec}(\xi_\mathrm{ex}^7)=\{\xi_\mathrm{ex}^{11},\xi_\mathrm{ex}^{12}\}=\mathrm{fdec}(\xi_\mathrm{ex})$,
  since those last two KLM sequences are normal.  The corresponding
  decomposition forest in depicted in \cref{fig-forest}.  Observe that
  the union provided in \cref{ex-lang} for~$L_{\xi_\mathrm{ex}}$
  corresponds exactly to the union of~$L_{\xi_\mathrm{ex}^{12}}$
  and~$L_{\xi_\mathrm{ex}^{11}}$.
\end{example}

\section{Complexity Upper Bounds}
\label{sec-up}
In this section, we derive upper bounds on the lengths of the branches
in a decomposition forest of a KLM sequence~$\xi_0$, from which we
can in turn provide upper bounds on the size of normal KLM sequences,
the length of small witnesses, the running time of the decomposition
algorithm, and the size of the full decomposition.  The idea is to
exploit the ranking function defined in \cref{sub-rank} in order to
bound how many decomposition steps can be performed along a branch of
a decomposition forest.  We rely for this on a so-called `length
function theorem' from~\cite{schmitz14} to bound the length of
descending sequences of ordinals.  Finally, we classify the running
time complexity using the `fast-growing' complexity classes defined
in~\citep{schmitz16}.  A general introduction to these techniques can
be found in~\citep{schmitz17}.

\subsection{Controlled Sequences of Ranks}

For the purposes of this section, it is more convenient to recast the
ranking function $\mathrm{rank}()$ on KLM sequences from \cref{sub-rank}
in terms of ordinals.  If $\rank{\xi}=(r_d,\dots,r_0)$, then we
associate to~$\xi$ the \emph{ordinal rank} in $\omega^{d+1}$ defined by
\begin{equation}\label{eq-rank}
  \alpha_\xi\eqdef \omega^{d}\cdot r_d+\omega^{d-1}\cdot
  r_{d-1}+\cdots+\omega^0\cdot r_0\;.
\end{equation}
This is just a reformulation, because
$\rank{\xi}<_\mathit{lex}\rank{\xi'}$ if and only if
$\alpha_\xi<\alpha_{\xi'}$.  Along a branch $\xi'_0,\xi_1,\xi_2,\dots$
of a decomposition forest for a KLM sequence~$\xi_0$, we see
therefore a descending sequence of ordinal ranks
\begin{equation}\label{eq-ranks}
  \alpha_{\xi'_0}>\alpha_{\xi_1}>\alpha_{\xi_2}>\cdots
\end{equation}

Though all descending sequences of ordinals are finite, we cannot
bound their lengths in general; e.g., $K+1>K>K-1>\cdots>0$ and
$\omega>K>K-1>\cdots>0$ are descending sequences of length $K+2$ for
all~$K$ in~$\+N$.  Nevertheless, a descending sequence of ordinal
ranks like~\eqref{eq-ranks}, found along a branch of a decomposition
forest, is not arbitrary, because the successive KLM sequences are
either $\xi'_0\in\mathrm{clean}(\xi_0)$ or the result of some
decomposition step, hence one cannot use an arbitrary~$K$ as in these
examples.

\subsubsection{Controlled Sequences of Ordinals}
The previous intuition is captured by the notion of \emph{controlled
  sequences}.  In general, for an ordinal $\alpha<\omega^\omega$ (like
the ordinal ranks defined by~\eqref{eq-rank}), let us write $\alpha$
in Cantor normal form as
$\alpha=\omega^{n}\cdot c_n+\cdots+\omega^0\cdot c_0$ with
$c_0,\dots,c_n$ and $n$ in~$\+N$, and define its \emph{size} as
$N\alpha\eqdef\max\{n,\max_{0\leq i\leq n}c_i\}$.  Thus, for the
ordinal rank~$\alpha_\xi$ defined in~\eqref{eq-rank} for a KLM
sequence~$\xi$ with $\rank{\xi}=(r_d,\dots,r_0)$,
\begin{align}\label{eq-norm}
  N\alpha_\xi&=\max\{d,\max_{0\leq i\leq d}r_i\}\;.
\end{align}

Let $n_0$ be a natural number in~$\+N$ and $h{:}\,\+N\to\+N$ a
monotone inflationary function, i.e., $x\leq h(x)$ and
$x\leq y$ implies $h(x)\leq h(y)$.
A sequence $\alpha_0,\alpha_1,\dots$
of ordinals below~$\omega^\omega$ is \emph{$(n_0,h)$-controlled} if,
for all $j$ in~$\+N$,
\begin{equation}\label{eq-ctrl}
  N\alpha_j\leq h^j(n_0)\;,
\end{equation}
i.e., the size of the $j$th ordinal $\alpha_j$ is bounded by the $j$th
iterate of~$h$ applied to~$n_0$; in particular, $N\alpha_0\leq n_0$
for the first element of the sequence.  Because for each~$n\in\+N$,
there are only finitely many ordinals below~$\omega^\omega$ of size at
most~$n$, the length of controlled descending sequences is
bounded~\citep[see, e.g.,][]{schmitz14}.  One can actually give a
precise bound on this length in terms of \emph{subrecursive
  functions}, whose definition we are about to recall.

\subsubsection{Subrecursive Functions}

Algorithms shown to terminate via an ordinal ranking function can have
a very high worst-case complexity.  In order to express such large
bounds, a convenient tool is found in subrecursive hierarchies, which
employ recursion over ordinal indices to define faster and faster
growing functions.  We define here two such hierarchies.

\paragraph{Fundamental Sequences}
A \emph{fundamental sequence} for a limit ordinal $\lambda$ is a
strictly ascending sequence $(\lambda(x))_{x<\omega}$ of ordinals
$\lambda(x)<\lambda$ with supremum~$\lambda$.  We use the standard
assignment of fundamental sequences to limit ordinals \ifams
$\lambda<\varepsilon_0$, where~$\varepsilon_0$ denotes the least
solution of~$x=\omega^x$.  For the purposes of this paper, it actually
suffices to consider the case \fi $\lambda\leq\omega^\omega$, defined
inductively by
\begin{align*}
  \omega^\omega(x)&\eqdef \omega^{x+1}\;,&
  (\beta+\omega^{k+1})(x)&\eqdef \beta+\omega^k\cdot(x+1)\;,
\end{align*}
where $\beta+\omega^{k+1}$ is in Cantor normal form.
This particular assignment satisfies, e.g., $0 < \lambda(x)
< \lambda(y)$ for all $x < y$. For instance, $\omega(x) = x + 1$ and
$(\omega^{3}+\omega^3+\omega)(x)=\omega^3+\omega^3+x+1$.

\paragraph{Hardy and Cicho\'n Hierarchies}
In the context of controlled sequences, the hierarchies of Hardy and
Cicho\'n turn out to be especially well-suited~\citep{cichon98}.  Let
$h{:}\,\+N\to\+N$ be a function.  For each such~$h$, the \emph{Hardy
hierarchy} $(h^\alpha)_{\alpha\leq\omega^\omega}$ and the \emph{Cicho\'n hierarchy}
$(h_\alpha)_{\alpha\leq\omega^\omega}$ relative to~$h$ are two families of functions
$h^\alpha,h_\alpha{:}\,\+N\to\+N$ defined by induction over~$\alpha$ by
\begin{align*}
  h^0(x)&\eqdef x\;,&
  h_0(x)&\eqdef 0\;,\\
  h^{\alpha+1}(x)&\eqdef h^\alpha(h(x))\;,
  &h_{\alpha+1}(x)&\eqdef1+h_\alpha(h(x))\;,\\
  h^\lambda(x)&\eqdef h^{\lambda(x)}(x)\;,&
  h_\lambda(x)&\eqdef h_{\lambda(x)}(x)\;.
\end{align*}
The Hardy functions are well-suited for expressing a large number of
iterations of the provided function~$h$.  For instance, $h^k$ for some
finite $k$ is simply the $k$th iterate of~$h$.  This intuition carries
over: $h^\alpha$ is a `transfinite' iteration of the function~$h$,
using a kind of diagonalisation in the fundamental sequences to handle
limit ordinals.  For instance, if we use the successor function $H(x)
= x+1$ as our function~$h$, we see that a first diagonalisation yields
$H^\omega(x) = H^{x+1}(x) = 2x+1$. The next diagonalisation occurs at
$H^{\omega\cdot 2}(x) = H^{\omega+x+1}(x)=H^\omega(2x + 1) = 4x +
3$. Fast-forwarding a bit, we get for instance a function of
exponential growth $H^{\omega^2}(x) = 2^{x+1} (x + 1) - 1$, and later
a non-elementary function $H^{\omega^3}$ akin to a tower of
exponentials, and a non primitive-recursive function
$H^{\omega^\omega}$ of Ackermannian growth.

In the following, we will use the fact that, if $h$ is monotone
inflationary, then so is~$h^\alpha$: if $x\leq y$, then
$x\leq h^\alpha(x)\leq h^\alpha(y)$.  Regarding the Cicho\'n
functions, if $h$ is monotone inflationary, then by induction
on~$\alpha$,
\begin{align}\label{eq-hardy}
h^\alpha(x)&\geq h_\alpha(x) + x\;.
\intertext{%
But the main interest of
Cicho\'n functions is that they capture how many iterations are
performed by Hardy functions~\citep{cichon98}:}
\label{eq-hardy-cichon}
  h^{h_\alpha(x)}(x)&=h^\alpha(x)\;.
\end{align}

\subsubsection{Length Function Theorem}
We can now state a `length function theorem' for controlled descending
sequences of ordinals.

\begin{theorem}[{\citep[Thm.~3.3]{schmitz14}}]\label{th-lft}
  Let $n_0\geq d+1$.  The maximal length of $(n_0,h)$-controlled
  descending sequences of ordinals in $\omega^{d+1}$ is
  $h_{\omega^{d+1}}(n_0)$.
\end{theorem}

Let us apply \cref{th-lft} to the descending sequences of ordinal
ranks from~\eqref{eq-ranks} found along a branch of a decomposition
forest of~$\xi_0$.  Observe that by~\eqref{eq-klm-size}
and~\eqref{eq-norm}, $N\alpha_\xi\leq\size{\xi}$ for any KLM
sequence~$\xi$.  Thus, by monotonicity, a sequence
like~\eqref{eq-ranks} is $(g(\size{\xi_0}),h)$-controlled, where~$g$
was defined in~\cref{lem:initialreduction} and~$h$ in
\cref{lem:inductivereduction}.  By \cref{th-lft} and because
$g(\size{\xi_0})\geq d+1$, the branches of a decomposition forest
for~$\xi_0$ are of length at most
\begin{align}
  L&\eqdef h_{\omega^{d+1}}(g(\size{\xi_0}))\;.\label{eq-L}
\intertext{%
In turn, by \eqref{eq-hardy} and~\eqref{eq-L},}
  L&\leq h^{\omega^{d+1}}(g(\size{\xi_0}))\;,\label{eq-length}
\intertext{and if $\xi$ is any KLM sequence labelling a node of a
     decomposition forest for~$\xi_0$, then
     by~\eqref{eq-hardy-cichon} and~\eqref{eq-L},}
  \size\xi&\leq
            h^L(g(\size{\xi_0}))=h^{\omega^{d+1}}(g(\size{\xi_0}))\;.\label{eq-size}
\end{align}

Consider now a VASS~$G$ of dimension~$d$ and two finite configurations
$\vec c_\mathit{in}$ and $\vec c_\mathit{out}$.  Then according
to~\eqref{eq-G-size} and~\eqref{eq-klm-size},
\begin{equation}\label{eq-sizes}\size{\vec c_\mathit{in}G \vec c_\mathit{out}}=2(d+1)^{d+1}(\size{G}+\norm{\vec c_\mathit{in}}{}+\norm{\vec c_\mathit{out}}{})\;.\end{equation}
Thus, by combining~\eqref{eq-size} with \cref{lem:normalnonempty}, we
obtain the following small witness property.
\begin{property}[Small Witness]\label{cor-witness}
  Let $G=(Q,q_\mathit{in},q_\mathit{out},T)$ be a VASS of
  dimension~$d$, $\vec{c}_{in}$ and $\vec{c}_{out}$ be two finite
  configurations in~$\+N^d$, and $n\eqdef 2(d+1)^{d+1}(\size{G}+\norm{\vec
    c_\mathit{in}}{}+\norm{\vec c_\mathit{out}}{})$.  If
  $\qin{\vec{c}_{in}}\step[G]{\sigma}\qout{\vec{c}_{out}}$ for
  some~$\sigma$, then there exists a word $\sigma'\in\vec{A}^*$ such
  that $\qin{\vec{c}_{in}}\step[G]{\sigma'}\qout{\vec{c}_{out}}$ and
  $$|\sigma'|\leq \ell\big(h^{\omega^{d+1}}(g(n))\big)\;,$$
  where $g$, $h$, and~$\ell$ are defined in
  \crefrange{lem:initialreduction}{lem:normalnonempty}.
\end{property}

\subsection{Fast-Growing Complexity}

We wish now to exploit the upper bounds
from~(\ref{eq-L}--\ref{eq-size}) and \cref{cor-witness} to provide
complexity upper bounds for the decomposition algorithm and the
reachability problem.  We will employ for this the \emph{fast-growing}
complexity classes defined in~\citep{schmitz16}.  This is an
ordinal-indexed hierarchy of complexity classes
$(\F\alpha)_{\alpha<\varepsilon_0}$, that uses the Hardy functions
$(H^\alpha)_\alpha$ relative to $H(x)\eqdef x+1$ as a standard against
which we can measure high complexities.

\subsubsection{Fast-Growing Complexity Classes}
Let us first define
\begin{align}\label{eq-FGH}
  \FGH\alpha&\eqdef\bigcup_{\beta<\omega^\alpha}\ComplexityFont{FDTIME}\big(H^\beta(n)\big)
  \intertext{%
    as the class of functions computed by deterministic Turing
    machines in time $O(H^\beta(n))$ for some $\beta<\omega^\alpha$.
    This captures for instance the class of Kalmar elementary
    functions as $\FGH 3$ and the class of primitive-recursive
    functions as $\FGH\omega$~\citep{lob70,wainer72}.  Then we let}
  \F\alpha&\eqdef\bigcup_{p\in\FGH\alpha}\ComplexityFont{DTIME}\big(H^{\omega^\alpha}\!(p(n))\big)\label{eq-F}
\end{align}
denote the class of decision problems solved by deterministic Turing
machines in time $O\big(H^{\omega^\alpha}\!(p(n))\big)$ for some
function~$p\in\FGH\alpha$.  The intuition behind this quantification
over~$p$ is that, just like e.g.\
$\EXP=\bigcup_{p\in\poly}\ComplexityFont{DTIME}\big(2^{p(n)}\big)$
quantifies over polynomial functions to provide enough `wiggle room' to account
for polynomial reductions, $\F\alpha$ is closed under $\FGH\alpha$
reductions~\citep[Thms.~4.7 and~4.8]{schmitz16}.

For instance, $\TOWER\eqdef\F 3$ defines the class of
problems that can be solved using computational resources bounded by a
tower of exponentials of elementary height in the size of the input,
$\bigcup_{k\in\+N}\F k$ is the class of primitive-recursive decision
problems, and $\ACK\eqdef\F\omega$ is the class
of problems that can be solved using computational resources bounded
by the Ackermann function applied to some primitive-recursive function
of the input size---here it does not matter for $\alpha>2$ whether we
are considering deterministic, nondeterministic, alternating, time, or
space bounds~\citep[Sec.~4.2.1]{schmitz16}.  \ifams See \cref{fig-fg} for a
depiction.\fi

\subsubsection{Complexity Upper Bounds}
Let us first observe that, by \cref{lem:inductivereduction}, the
branching degree~$|\mathrm{dec}(\xi)|$ of a node labelled by~$\xi$ in
a decomposition forest for~$\xi_0$ is exponential
in~$h(\size{\xi})$, thus elementary in~$\size{\xi}$.  Furthermore, by
\cref{lem:initialreduction}, the number~$|\mathrm{clean}(\xi_0)|$ of
initial clean KLM sequences is exponential in~$g(\size{\xi_0})$, thus
elementary in~$\size{\xi_0}$.  Thus, by~\eqref{eq-size}, the size of
the entire forest---i.e., the number of decomposition steps performed
by the decomposition algorithm---is also elementary in
$h^{\omega^{d+1}}(g(\size{\xi_0}))$.  Finally, still by
\cref{lem:inductivereduction}, each decomposition step on a KLM
sequence~$\xi$ can be performed in time elementary in~$\size\xi$,
hence the entire decomposition forest can be computed in time
elementary in $h^{\omega^{d+1}}(g(\size{\xi_0}))$.

\begin{lemma}\label{lem-fdec}
  Given a KLM sequence $\xi_0$ of dimension~$d$, we can compute
  $\mathrm{fdec}(\xi_0)$ in time $e(h^{\omega^{d+1}}(g(\size{\xi_0})))$
  for some fixed elementary function~$e$ and where $g$ and $h$ are defined in
  \cref{lem:initialreduction,lem:inductivereduction}.
\end{lemma}

Consider an instance of the \nameref{reach} problem, namely a VASS~$G$
of dimension~$d$ and two finite configurations $\vec c_\mathit{in}$
and $\vec c_\mathit{out}$, and let
$\xi_0\eqdef(\vec c_\mathit{in}G\vec c_\mathit{out})$.  Then
$\mathrm{fdec}(\xi_0)=\emptyset$ if and only
$\qin{\vec c_\mathit{in}}\step[G]{\ast}\qout{c_\mathit{out}}$, where
by \eqref{eq-sizes}, $\size{\xi_0}$ is elementary in the size of the
instance.  Let us examine the bound $e(h^{\omega^{d+1}}(g(\size{\xi_0})))$
from \cref{lem-fdec} and express it in the form of~\eqref{eq-F}.  The
innermost~$g$ function composed with the blow-up incurred
by~\eqref{eq-sizes} is a fixed elementary function in~$\FGH{<3}$, thus
is captured by the quantification over $p\in\FGH{<\alpha}$ for all
$\alpha\geq 3$.  The inner~$h$ function is also fixed and
in~$\FGH{<3}$, and \citep[Thm.~4.2]{schmitz13} allows to
over-approximate~$h^{\omega^{d+1}}$ in terms of~$H^{\omega^{d+4}}$.
Finally, the outermost function~$e$ is also fixed and in~$\FGH{<3}$,
and \citep[Lem.~4.6]{schmitz13} shows how to `shift' it into the
innermost position.

\begin{theorem}[Upper Bound]\label{th-up}
  \nameref{reach} is in \ACK, and in $\F{d+4}$ if the dimension~$d$ is
  fixed.
\end{theorem}

\subsubsection{Combinatorial Algorithm} 
An alternative proof of \cref{th-up} could also exploit the following
\emph{combinatorial algorithm}.  By \cref{cor-witness}, if
$\qin{\vec c_\mathit{in}}\step[G]{\sigma}\qout{c_\mathit{out}}$ for
some~$\sigma$, then there is a small witness~$\sigma'$ of length at
most
$\ell(h^{\omega^{d+1}}(g(n)))$.
It suffices therefore to compute this upper bound---which can be
performed in time elementary in the
bound~\citep[Thm.~5.1]{schmitz13}---, and to enumerate the paths
in~$G$ of length up to that bound until we find a witness or exhaust
the search space.

\section{Application: Downward Language Inclusion}
\label{sec-apps}
The $\ACK=\F\omega$ upper bound provided by \cref{th-up} for
the \nameref{reach} problem is still quite far from the currently best
known lower bound, which is $\TOWER=\F3$
hardness~\citep{czerwinski19}.  As mentioned in the introduction, this
upper bound is nevertheless rather tight as far as the decomposition
algorithm is concerned.  In this section, we illustrate the usefulness
of our new upper bound for another decision problem.

\paragraph{Labelled VASSes}
A \emph{labelled} VASS $(G,\Sigma,\lambda)$ is a
VASS~$G=(Q,q_\mathit{in},q_\mathit{out},T)$ of dimension~$d$ together
with a finite alphabet~$\Sigma$ and labelling
function~$\lambda{:}\,T\to\Sigma\cup\{\varepsilon\}$, which is lifted
homomorphically to a function $T^\ast\to\Sigma^\ast$.
We overload the notations for step relations by writing $p(\vec
x)\step[G]{w}q(\vec y)$ if there exists a path $\pi\in T^\ast$
from~$p$ to~$q$ labelled by~$\sigma$ such that $p(\vec
x)\step[G]{\sigma}q(\vec y)$ and $\lambda(\pi)=w\in\Sigma^\ast$.
Given two finite configurations $\vec c_\mathit{in}$ and $\vec
c_\mathit{out}$ in~$\+N^d$, its \emph{labelled language} is
\begin{multline*}
  L_\lambda(\vec c_\mathit{in},G,\vec c_\mathit{out})\eqdef\{w\in\Sigma^\ast\mid\qin{\vec c_\mathit{in}}\step[G]{w}\qout{\vec c_\mathit{out}}\}\;.
\end{multline*}

\paragraph{Downward-Closures}

For two finite words $u$ and $v$ in~$\Sigma^\ast$, we say
that~$u$ \emph{embeds} into~$v$, denoted $u\leq_\ast v$, if
$u=a_1\cdots a_k$ and $v=v_0a_1v_1a_2\cdots a_kv_k$ for some
$a_1,\dots,a_k\in\Sigma$ and $v_0,\dots,v_k\in\Sigma^\ast$.  In other
words, $u$ embeds into~$v$ if we can obtain $u$ from~$v$ by `dropping'
some letters from~$v$; for instance, $bca\leq_\ast aabacba$.  For a
language $L\subseteq\Sigma^\ast$, its \emph{downward-closure} is
${\downarrow}L\eqdef\{u\in\Sigma^\ast\mid\exists v\in
L\mathbin.u\leq_\ast v\}$.  A consequence of Higman's Lemma also known
as Haine's Theorem is that, for any $L\subseteq\Sigma^\ast$,
${\downarrow}L$ is a regular language.

\begin{example}\label{ex-down}
  Let us consider again the VASS~$G_\mathrm{ex}$ from \cref{ex-vass},
  along with the alphabet $\Sigma\eqdef\{a_j\mid 1\leq j\leq 9\}$ and
  the labelling function defined by~$\lambda(t_j)\eqdef a_j$ for all
  $1\leq j\leq 9$.  Then
  ${\downarrow}L_\lambda((0,0,2),G_\mathrm{ex},(1,1,0))$ is the
  language denoted by the regular expression
  \begin{equation*}
    a_1^\ast(a_3+\varepsilon)a_6^\ast(a_7+\varepsilon)a_8^\ast(a_9+\varepsilon)(a_6+\varepsilon)\;.\qedhere
  \end{equation*}
\end{example}

\paragraph{Downward Language Inclusion}

We are interested in this section in the following decision problem.
\begin{problem}[VASS downward language inclusion]\label{incl}
  \hfill\vspace*{-1ex}\begin{description}[\IEEEsetlabelwidth{question}] \item[input]
  Two labelled VASSes $(G,\Sigma,\lambda)$ and $(G',\Sigma,\lambda')$
  and four finite configurations $\vec c_\mathit{in}$ and $\vec
  c_\mathit{out}$ of~$G$ and $\vec c'_\mathit{in}$ and $\vec
  c'_\mathit{out}$ of~$G'$.
  \item[question] Is ${\downarrow}L_{\lambda}(\vec
  c_\mathit{in},G,\vec
  c_\mathit{out})\subseteq{\downarrow}L_{\lambda'}(\vec
  c'_\mathit{in},G',\vec c'_\mathit{out})$?
\end{description}
\end{problem}

Now, by Haine's Theorem,
${\downarrow}L_{\lambda}(\vec c_\mathit{in},G,\vec c_\mathit{out})$ is
regular for any labelled VASS.  However, that does not necessarily
mean that one can actually compute a finite automaton $\?A$ such that
$L(\?A)={\downarrow}L_{\lambda}(\vec c_\mathit{in},G,\vec
c_\mathit{out})$
from the labelled VASS and configurations.  Nevertheless,
\citet*[Prop.~1]{habermehl10} show that, given a full decomposition
$\mathrm{fdec}(\xi_0)$ of the KLM sequence
$\xi_0=(\vec c_\mathit{in}G\vec c_\mathit{out})$, one can construct
such a finite automaton in logarithmic space\footnote{The result of
  \citeauthor{habermehl10} is stated in terms of a full decomposition
  constructed by \citeauthor{lambert92}'s algorithm, but the
  adaptation to our full decomposition is mostly
  straightforward.}---as a hint, the reader might see some resemblance
between the regular expression of \cref{ex-down} and the full
decomposition
$\mathrm{fdec}(\xi_\mathrm{ex})=\{\xi_\mathrm{ex}^{11},\xi_\mathrm{ex}^{12}\}$
from \cref{ex-pump}.  Since the inclusion problem for two regular
languages represented by finite automata is in \PSPACE,
\cref{lem-fdec} entails the following.

\begin{corollary}[{of \citep[Prop.~1]{habermehl10}}]\label{cor-incl-up}
  The \nameref{incl} problem is in \ACK, and in $\F{d+4}$ if the
  dimension of the labelled VASSes is fixed to~$d$.
\end{corollary}

\paragraph{Lower Bounds}

The computational and the descriptional complexity of computing
downward-closures of languages is rather well
studied~\citep[e.g.,][]{zetzsche15}.  In the case of labelled VASS
languages, \citet[Thm.~10]{atig17} show that there exists a family of
labelled VASSes such that any finite automaton $\?A$ such that
$L(\?A)={\downarrow}L_{\lambda}(\vec c_\mathit{in},G,\vec
c_\mathit{out})$
requires a number of states Ackermannian in the size of~$G$.  A
stronger result was shown by~\citet[Cor.~17]{zetzsche16}, which bars
any alternative algorithm for the \nameref{incl} problem from
performing significantly better than \cref{cor-incl-up}.
\begin{theorem}[{\citep[Cor.~17]{zetzsche16}}]
  The \nameref{incl} problem is \ACK-hard, and $\F d$-hard if the
  dimension of the labelled VASSes is fixed to~$d\geq 3$.
\end{theorem}
\begin{proof}
  The lower bound in the general case is stated
  in~\citep[Cor.~17]{zetzsche16}.  Regarding the case in fixed
  dimension~$d\geq 1$, \citep[Thm.~15]{zetzsche16} shows how to derive
  $\F d$-hardness provided we can `weakly implement' the Ackermann
  function $A_d(x)$ with a VASS of dimension~$d$ and size polynomial
  in~$d$, where $A_d(x)$ is defined inductively by $A_1(x)\eqdef 2x$
  and $A_{d+1}(x)\eqdef A_d^x(1)$.  The existence of such weak
  implementations is well-known; see for
  instance~\citep[Sec.~4.2.2]{schmitz17}.
\end{proof}

\section{Concluding Remarks}
\label{sec-concl}
We have proven that a refinement of the decomposition algorithms of
\citet{mayr81}, \citet{kosaraju82}, and \citet{lambert92} runs in
Ackermannian time, and in primitive-recursive time for VASSes of fixed
dimension.  In turn, this provides respectively \ACK\ and $\F{d+4}$
upper bounds for both the \nameref{reach} and the \nameref{incl}
problems.  While the former only needs to find \emph{some} normal KLM
sequence in a decomposition forest and is only known to be
\TOWER-hard~\citep{czerwinski19}, the latter essentially requires to
construct a full decomposition and was already known to be
\ACK-hard~\citep{zetzsche16}, and therefore \ACK-complete by our
results.  Thus it is unclear at the moment whether a better algorithm
for \nameref{reach} might exist.

Another line for future research is the complexity of \nameref{reach}
in fixed dimension.  With a binary encoding, the problem is
\NP-complete in dimension one~\citep{haase09} and \PSPACE-complete in
dimension two~\citep{blondin15}; with a unary encoding, both are
\NL-complete~\citep{englert16}.  In dimension three and above, our
$\F{d+4}$ bound is currently the best known upper bound, but we expect
that this could be refined further.

\section*{Acknowledgements}

The authors thank the LICS~2019 reviewers for their thorough reviews.
Work partially funded by
\href{http://bravas.labri.fr/}{ANR-17-CE40-0028~\textsc{Bra\!VAS}}.
The second author is also partially funded by \ifams\relax\else
Institut Universitaire de France and
\fi\href{http://projects.lsv.fr/prodaq}{ANR-14-CE28-0005
  \textsc{prodaq}}.

\bibliographystyle{abbrvnat}
\bibliography{conferences,journalsabbr,references}

\appendix
\section{Pumpability}
\label{sec-pump}
In this section we provide the details of the proof of
\cref{lem:pump}.

\subsection{Rackoff Extraction}
We start by proving the following result inspired by
\citet{rackoff78}.

\begin{lemma}\label{lem:rack}
  Let us assume that 
  $$q_0(\vec{c}_0)\step[G]{\vec{a}_1}q_1(\vec{c}_1)\cdots\step[G]{\vec{a}_k}q_k(\vec{c}_k)\;.$$
  Let $ C\geq\size{G}$ and assume that for every $i\in\{1,\ldots,d\}$ there exists $j\in\{0,\ldots,k\}$
  such that
  $$\vec{c}_j(i)\geq  C^{1+d^d}\;.$$
  In that case there exists a configuration
  $\vec{c}\in\setN_\omega^d$ such that
  $\vec{c}(i)\geq  C-\size{G}$ for every
  $i\in\{1,\ldots,d\}$, and a word $\sigma$ such that $|\sigma|<
   C^{(d+1)^{d+1}}$ and $q_0(\vec{c}_0)\step[G]{\sigma}q_k(\vec{c})$.
\end{lemma}

The previous lemma is a direct consequence of the following statement.
\begin{lemma}\label{lem:fullrack}
  Let us assume that 
  $$q_0(\vec{c}_0)\step[G]{\vec{a}_1}q_1(\vec{c}_1)\cdots\step[G]{\vec{a}_k}q_k(\vec{c}_k)\;.$$ 
  Let $ C\geq \size{G}$ be such that for every $i\in\{1,\ldots,d\}$ there exists $j\in\{0,\ldots,k\}$
  such that
  $$\vec{c}_j(i)\geq  C^{1+n^n}$$
  where $n\eqdef|\{i\mid \vec{c}_0(i)\in\setN\}|$.
  Then, there exists a
  configuration $\vec{c}$ and a word $\sigma$ such that
  $q_0(\vec{c}_0)\step[G]{\sigma}q_k(\vec{c})$,
  $\vec{c}(i)\geq  C-\size{G}$ for every $i$, and
  $$|\sigma|<  C^{(n+1)^{n+1}}\;.$$
\end{lemma}
\begin{proof}
  Let $T$ be the set of transitions of~$G$. Notice that if~$T$ is
  empty, the lemma is immediate.  So, we can assume
  that~$T\neq\emptyset$, and in particular that $\size{G}\geq 2$.

  \medskip

  We prove the lemma by induction on~$n$.  Naturally, for the base
  case $n=0$, then $\vec{c}_0=\vec{\omega}$ and the proof is immediate
  by selecting for $\sigma$ the label of a simple path from~$q_0$
  to~$q_k$; notice that $|\sigma|<|Q|\leq C$.  For the induction step,
  let us assume that the lemma holds as soon as the cardinal of
  $H\eqdef\{i\in\{1,\ldots,d\} \mid \vec{c}_0(i)\not=\omega\}$ is
  strictly bounded by $n\geq 1$, and let us consider an instance of
  the lemma such that $|H|=n$.  We have by
  assumption \begin{equation}q_0(\vec{c}_0)\step[G]{\vec{a}_1}q_1(\vec{c}_1)\cdots\step[G]{\vec{a}_k}q_k(\vec{c}_k)\;.\end{equation}
  Assume that for every $i\in\{1,\ldots,d\}$ there exists
  $j\in\{0,\ldots,k\}$ such that $\vec{c}_j(i)\geq C^{1+n^n}$.
  
  \medskip 
 
  Notice that for every $i\in H$, there exists a minimal
  $k_i\in\{0,\ldots,k\}$ such that $\vec{c}_{k_i}(i)\geq C^{1+n^n}$.
  Let $\hat{k}\eqdef\min\{k_i \mid i\in H\}$. Observe that for every
  $j\in\{0,\ldots,\hat{k}-1\}$ and for every~$i\in H$, by minimality
  of~$k_i$ we deduce that $\vec{c}_j(i)< C^{1+n^n}$. It follows that
  the cardinal of the set $\{q_j(\vec{c}_j) \mid 0\leq j<\hat{k}\}$ is
  bounded by $|Q|\cdot ( C^{1+n^{n}})^n\leq C^{1+n+n^{n+1}}$.  By
  removing cycles that occur on the execution from~$q_0$
  to~$q_{\hat{k}}$, we can assume without loss of generality that
  $\hat{k}\leq C^{1+n+n^{n+1}}$.

  \medskip

  Let $I\eqdef \{i\in H \mid k_i>\hat{k}\}$.
  Observe that $\hat{n}\eqdef|I|$ satisfies $\hat{n}<n$. Let us
  define $\vec{x}_j\eqdef\vec{c}_j|_{I}$. Then
  \begin{equation}q_{\hat{k}}(\vec{x}_{\hat{k}})\xrightarrow[G]{\vec{a}_{\hat{k}+1}}q_{\hat{k}+1}(\vec{x}_{\hat{k}+1})\cdots\xrightarrow[G]{\vec{a}_k}q_k(\vec{x}_k)\;.\end{equation}
  Moreover, for every $i\in\{1,\ldots,d\}$ there exists
  $j\in\{\hat{k},\ldots,k\}$ such that
  $\vec{x}_j(i)\geq C^{1+{\hat{n}}^{\hat{n}}}$.  In fact, if
  $i\not\in I$ then
  $\vec{x}_{\hat{k}}(i)=\omega\geq C^{1+{\hat{n}}^{\hat{n}}}$, and if
  $i\in I$, and since $k_i>\hat{k}$, there exists
  $j\in\{\hat{k}+1,\ldots,k\}$ such that
  $\vec{x}_j(i)\geq C^{1+n^n}\geq C^{1+{\hat{n}}^{\hat{n}}}$. By
  induction hypothesis, there exists a
  configuration~$\hat{\vec{c}}$ and a word~$\hat{\sigma}$ such that
  $q_{\hat{k}}(\vec{x}_{\hat{k}})\step[G]{\hat{\sigma}}q_k(\hat{\vec{c}})$,
  $\hat{\vec{c}}(i)\geq C-\size{G}$ for every~$i$, and
  $|\hat{\sigma}|< C^{{(\hat{n}+1)}^{\hat{n}+1}}\leq C^{n^n}$.

  \medskip

  Let $\sigma\eqdef
  \vec{a}_1\ldots\vec{a}_{\hat{k}}\hat{\sigma}$.
  Observe that $|\sigma|< C^{1+n+n^{n+1}}+ C^{n^n}$. Let us prove that
  $C^{1+n+n^{n+1}}+C^{n^n}\leq C^{(n+1)^{n+1}}$. Since $C\geq 2$, it is
    sufficient to prove that $(1+n+n^{n+1})+n^n\leq (n+1)^{n+1}$. If
    $n=1$ the inequality is trivial. Otherwise, 
    the first two elements of the binomial decomposition of $(n+1)^{n+1}$ provide
    $(n+1)^{n+1}\geq n^{n+1}+(n+1)n^n$.  Moreover $(n+1)n^n\geq
    n^n+n^{n+1}\geq n^n+2n\geq n^n+n+1$.  We have proven the inequality.

  \medskip
  
  Let $\setZ_\omega\eqdef\setZ\uplus\{\omega\}$, and 
  let us introduce for every prefix $u$ of $\hat{\sigma}$ the vector
  $\vec{z}_u$ in $\setZ_\omega^d$ defined by
  $\vec{z}_u\eqdef\vec{c}_{\hat{k}}+\Delta(u)$. 
  Observe that
  $\vec{z}_u|_{I}=\vec{x}_{\hat{k}}+\Delta(u)$.
  Since
  $q_{\hat{k}}(\vec{x}_{\hat{k}})\step[G]{\hat{\sigma}}q_k(\hat{\vec{c}})$,
  there exists a configuration $q_u(\hat{\vec{z}}_u)$ such that
  $q_{\hat{k}}(\vec{x}_{\hat{k}})\step[G]{u}q_u(\hat{\vec{z}}_u)$.
  In particular $\hat{\vec{z}}_u=\vec{x}_{\hat{k}}+\Delta(u)$, and we
  have proven that $\vec{z}_u|_{I}=\hat{\vec{z}}_u$. It
  follows that $\vec{z}_u(i)\geq 0$ for every $i\in I$. Notice
  that for every $i\not\in H$, we have $\vec{z}_u(i)=\omega\geq
  0$. Moreover, for every $i\in H\backslash I$, we have
  $\vec{z}_u(i)=\vec{c}_{\hat{k}}(i)+\Delta(u)(i)\geq
   C^{1+n^n}-\size{G}\cdot|u|\geq  C^{1+n^n}-\size{G}\cdot
   C^{n^n}= C^{n^n}\cdot ( C-\size{G})\geq
   C-\size{G}$.
  Hence
  $\vec{z}_u(i)\geq
   C-\size{G}$.
  We have proven that $\vec{z}_u\in\setN_\omega^d$. In particular, it
  follows that
  $q_{\hat{k}}(\vec{c}_{\hat{k}})\step[G]{\hat{\sigma}}q_k(\vec{c})$
  where $\vec{c}=\vec{z}_{\hat{\sigma}}$. Notice that we have
  $\vec{c}(i)\geq  C-\size{G}$ for every~$i$. In fact, if $i\in
  H\backslash I$ then $\vec{c}(i)=\vec{z}_{\hat{\sigma}}(i)$, if
  $i\in I$ then $\vec{c}(i)=\hat{\vec{c}}(i)$, and if $i\not\in
  H$ then $\vec{c}(i)=\omega$. Finally, as
  $q_0(\vec{c}_0)\step{\sigma_0}q_{\hat{k}}(\vec{c}_{\hat{k}})$
  we deduce that $q_0(\vec{c}_0)\step[G]{\sigma}q_k(\vec{c})$
  and we have proven the induction.
\end{proof}

\subsection{Unfoldings}

Recall that $\setN_B\eqdef\{0,\ldots,B-1,\omega\}$ for any
$B\in\setN$, and that if $i\in\{1,\ldots,d\}$, $r\in\setN_B$, and
$\vec{x}(i)\in\setN_B$, then the \emph{forward $(i,B,r)$-unfolding} of
a KLM triple $\vec{x}G\vec{y}$ is the KLM triple $\vec{x}G'\vec{y}$
where
$G'\eqdef(Q\times\setN_B,(q_\mathit{in},\vec{x}(i)),(q_\mathit{out},r),T')$
and~$T'$ is the set of transitions $((p,m),\vec{a},(q,n))$ where
$(p,\vec{a},q)\in T$ and $m,n\in\setN_B$ satisfy $n=m+\vec{a}(i)$ or
$(n=\omega\wedge m+\vec{a}(i)\geq B)$, and such that $m=\omega$
implies $q\neq q_\mathit{in}$.  The \emph{backward
  $(i,B,r)$-unfolding} is defined symmetrically.  \ifams\relax\else The condition that
$m=\omega$ must imply $q\neq q_\mathit{in}$ is central for proving the
following lemma.
\rankred*

\fi

\cycleplus*
\begin{proof}
  The inclusion of the right hand side to the left hand side is
  immediate. Let us prove the other inclusion.

  \medskip
  
  Let $\sigma\eqdef\vec{a}_1\ldots\vec{a}_k$ be a word
  in~$L_\xi$. Then there exists a sequence
  $q_0(\vec{m}_0),\ldots,q_k(\vec{m}_k)$ of state-configurations such
  that $q_0=q_\mathit{in}$, $q_k=q_\mathit{out}$,
  $\vec{m}_0,\vec{m}_k\in\setN^d$, $\vec{m}_0\sqsubseteq \vec{x}$,
  $\vec{m}_k\sqsubseteq \vec{y}$, and such that
  $$q_0(\vec{m}_0)\step[G]{\vec{a}_1}q_1(\vec{m}_1)\cdots
  \step[G]{\vec{a}_k}q_k(\vec{m}_k)\,.$$
  
  \medskip

  Without loss of generality, by restricting the set of states of $G$
  to~$\{q_0,\ldots,q_k\}$, we can assume that every state of~$G$ is
  visited. Let us prove that for every component~$i$ fixed by~$G$ such
  that $\Facc{G}{\vec{x}}(i)\not=\omega$, there exists a function
  $f_i{:}\,Q\rightarrow\setN$ such that
  $f(q_\mathit{in})=\Facc{G}{\vec{x}}(i)$ and such that
  $f_i(q)=f_i(p)+\vec{a}(i)$ for every $(p,\vec{a},q)\in T$.
  Since~$i$ is fixed by~$G$, there exists a function
  $f_i{:}\,Q\rightarrow\setZ$ such that $f_i(q)=f_i(p)+\vec{a}(i)$ for
  every $(p,\vec{a},q)\in T$.  As $\Facc{G}{\vec{x}}(i)\in\setN$, by
  translating~$f_i$ we can assume that
  $f_i(q_\mathit{in})=\Facc{G}{\vec{x}}(i)$. Notice that
  $\vec{m}_0\sqsubseteq \vec{x}\sqsubseteq \Facc{G}{\vec{x}}$. It
  follows that $\vec{m}_0(i)=\Facc{G}{\vec{x}}(i)$ and in particular
  that $f_i(q_0)=\vec{m}_0(i)$. Because
  $q_{j-1}(\vec{m}_{j-1})\step[G]{\vec{a}_j}q_j(\vec{m}_j)$, we deduce
  by induction on~$j$ that $f_i(q_j)=\vec{m}_j(i)$ for every
  $0\leq j\leq k$. As $Q=\{q_0,\ldots,q_k\}$, we deduce that
  $f_i(q)\in\setN$ for every $q\in Q$.

  \medskip

  Observe that, if there exists~$i\in I$ such that for every
  $j\in\{0,\ldots,k\}$ we have $\vec{m}_j(i)< B$, then
  $\sigma\in L_{\xi_r}$ where $r\eqdef\vec{m}_k(i)$.  Thus, we can
  assume that for every~$i\in I$, there exists $j\in\{0,\ldots,k\}$
  such that $\vec{m}_j(i)\geq B$. Let~$k'$ be the minimal natural
  number such that, for every $i\in I$, there exists
  $j\in\{0,\ldots,k'\}$ such that $\vec{m}_j(i)\geq B$.  Since
  $\vec{m}_0(i)< B$ for every~$i\in I$, it follows that~$k'\geq 1$.
  By minimality of~$k'$, there exists~$i\in I$ such that for every
  $j\in\{0,\ldots,k'-1\}$ we have $\vec{m}_j(i)< B$. Observe that if
  $q_\mathit{in}\not\in \{q_{k'},\ldots,q_k\}$, we deduce that
  $\sigma\in L_{\xi_\omega}$.  So, it just remains to prove that
  $q_\mathit{in}\not\in \{q_{k'},\ldots,q_k\}$.

  \medskip

  Assume by contradiction that
  $q_\mathit{in}\in\{q_{k'},\ldots,q_k\}$.  Since
  $\vec{m}_0\sqsubseteq \vec{x}\sqsubseteq \Facc{G}{\vec{x}}$ and
  $I\subseteq \{i\mid \Facc{G}{\vec{x}}(i)\in\setN\}$, it follows that
  $\vec{m}_0|_I=\Facc{G}{\vec{x}}|_I$.  Let
  $\vec{c}_j\eqdef\vec{m}_j|_I$ for every $0\leq j\leq k$. Notice that
  we have
  \begin{equation}q_\mathit{in}(\Facc{G}{\vec{x}}|_I)=q_0(\vec{c}_0)\step[G]{\vec{a}_1}q_1(\vec{c}_1)\cdots
  \step[G]{\vec{a}_{k'}}q_{k'}(\vec{c}_{k'})\,.\end{equation}
  For every $i\in\{1,\ldots,d\}$, there exists
  $j\in\{0,\ldots,k'\}$ such that $\vec{c}_j(i)\geq
   C^{1+d^d}$, where
  $C\eqdef\norm{\vec{x}}{}+2\size{G}$.
  In fact, if $i\not\in I$ then $\vec{c}_0(i)=\omega$,
  and if $i\in I$ then there exists $j\in\{0,\ldots,k'\}$
  such that $\vec{m}_j(i)\geq B$; From $\vec{c}_j(i)=\vec{m}_j(i)$
  we are done. \Cref{lem:rack} shows that there exists a
  configuration~$\vec{\bar{x}}$ and a word~$u'$ such that
  $q_0(\vec{c}_0)\step[G]{u'}q_{k'}(\vec{\bar{x}})$ and such that
  $\bar{\vec{x}}(i)\geq  C-\size{G}$ for every $1\leq i\leq
  d$. 

  \medskip
  
  Since $q_\mathit{in}\in \{q_{k'},\ldots,q_k\}$, we deduce that there
  exists a path in~$G$ from~$q_{k'}$ to~$q_\mathit{in}$.  We can
  consider a simple path of that form.  Let~$u''$ be the label of that
  path. Because $|u''|<|Q|$, we know that for every prefix $v$ of $u''$
  we have $\Delta(v)(i)> -\size{G}$. It follows that
  $\bar{\vec{x}}+\Delta(v)\geq \vec{0}$.  We have proven that
  $q_{k'}(\vec{\bar{x}})\step[G]{u''}q_\mathit{in}(\vec{z})$ where
  $\vec{z}$ satisfies $\vec{z}(i)> C-2\size{G}$ for every
  $1\leq i\leq d$. As $q_0(\vec{c}_0)=q_\mathit{in}(\vec{x}|_I)$, we
  also know that
  $q_\mathit{in}(\Facc{G}{\vec{x}}|_I)\step[G]{u}q_\mathit{in}(\vec{z})$
  where $u=u'u''$. Since $\vec{z}(i)>\vec{x}(i)=\Facc{G}{\vec{x}}(i)$ for
  every~$i\in I$, we deduce that $\Delta(u)(i)>0$ for every~$i\in I$.

  \medskip
  
  Finally, let us prove that there exists a configuration
  $\vec{x}'\geq \Facc{G}{\vec{x}}$ such that
  $q_\mathit{in}(\Facc{G}{\vec{x}})\step[G]{u}q_\mathit{in}(\vec{x}')$,
  and such that $\vec{x}'(i)>\Facc{G}{\vec{x}}(i)$ for every~$i\in I$.
  Let~$v$ be a prefix of~$u$ and let us first prove that
  $\Facc{G}{\vec{x}}(i)+\Delta(v)(i)\geq 0$ for every $1\leq i\leq d$.
  Note that if~$i\in I$, then because
  $q_\mathit{in}(\Facc{G}{\vec{x}}|_I)\step[G]{u}q_\mathit{in}(\vec{z})$,
  we get the property. If $\Facc{G}{\vec{x}}(i)=\omega$ then
  $\Facc{G}{\vec{x}}(i)+\Delta(v)(i)\geq 0$ is immediate. Therefore we
  can assume that $\Facc{G}{\vec{x}}(i)\not=\omega$ and~$i\not\in I$.
  In that case~$i$ is fixed by~$G$. We have seen in that case that
  there exists a function $f_i{:}\,Q\rightarrow\setN$ such that
  $f_i(q_\mathit{in})=\Facc{G}{\vec{x}}(i)$ and
  $f_i(q)=f_i(p)+\vec{a}(i)$ for every $(p,\vec{a},q)\in T$. We deduce
  that $\Facc{G}{\vec{x}}(i)+\Delta(u)(i)=f_i(q)\geq 0$ where~$q$ is any
  state reachable from $q_\mathit{in}$ by a path labelled by~$v$. It
  follows that
  $q_\mathit{in}(\Facc{G}{\vec{x}})\step[G]{u}q_\mathit{in}(\vec{x}')$
  for $\vec{x}'\eqdef\Facc{G}{\vec{x}}+\Delta(u)$. Notice that for
  every $i\not\in I$ this shows that $\vec{x}'(i)=\Facc{G}{\vec{x}}$,
  and that for every $i\in I$ and because $\Delta(u)(i)>0$, we have
  $\vec{x}'(i)>\Facc{G}{\vec{x}}(i)$.

  \medskip

  By monotony, notice that there exist a word~$\sigma$ and a
  configuration~$\vec{x}''\geq \vec{x}$ such that
  $q_\mathit{in}(\vec{x})\step[G]{\sigma}q_\mathit{in}(\vec{x}'')$ and
  such that $\vec{x}''(i)>\vec{x}(i)$ for every~$i$ such that
  $\vec{x}(i)\in\setN$ and $\Facc{G}{\vec{x}}(i)=\omega$. It follows
  that, for every $n\in\setN$ large enough, there exists a
  configuration $\vec{x}_n\geq \vec{x}$ such that
  $q_\mathit{in}(\vec{x})\step[G]{\sigma^nu}q_\mathit{in}(\vec{x}_n)$. Notice
  that for $n$ large enough, we have $\vec{x}_n\geq
  \vec{x}$. Moreover, we have $\vec{x}_n(i)>\vec{x}(i)$ for
  every~$i\in I$. Hence $\Facc{G}{\vec{x}}(i)=\omega$ for every
  $i\in I$ and we get a contradiction.
\end{proof}

\section{Unbounded Equations}
\label{sec-trans}

We recall some elements of linear algebra adapted from~\cite{pottier91}.
\begin{lemma}[{corollary of~\cite[Thm.~1]{pottier91}}]\label{lem:pottier}
  Let $(a_{i,j})_{\stackrel{1\leq i\leq m}{\text{\tiny$1{\leq} j{\leq} n$}}}$ be a
  sequence of integers and let $\vec{c}\in\setZ^m$ and let us define two
  sets $\vec{X}$ and $\vec{X}_0$ by
  \begin{align*}
    \vec{X}&\eqdef\biggl\{\vec{x}\in
    \setN^n\mathrel{\Big|}\bigwedge_{i=1}^m\sum_{j=1}^na_{i,j}\vec{x}(j)=\vec{c}(i)\biggr\}\;,&
    \vec{X}_0&\eqdef \biggl\{\vec{x}\in
    \setN^n\mathrel{\Big|}\bigwedge_{i=1}^m\sum_{j=1}^na_{i,j}\vec{x}(j)=0\biggr\}\;.
  \end{align*}
  Then every vector in~$\vec{X}$ can be decomposed as the sum of a
  vector~$\vec{x}$ in~$\vec{X}$ and a finite sum of
  vectors~$\vec{x}_0$ in~$\vec{X}_0$ such that:
  \begin{align*}
    \norm{\vec{x}}{}&\leq \norm{\vec{c}}{}\cdot
    (2+\max_{1\leq i\leq m}\sum_{j=1}^n|a_{i,j}|)^m\;,&
    \norm{\vec{x}_0}{}&\leq
    (2+\max_{1\leq i\leq m}\sum_{j=1}^n|a_{i,j}|)^m\;.
  \end{align*}
\end{lemma}
\begin{proof}
  Observe that if $\vec{c}(i)<0$ for some $1\leq i\leq m$,  by
  replacing~$\vec{c}(i)$ by~$-\vec{c}(i)$ and~$a_{i,j}$ by~$-a_{i,j}$
  for every $1\leq j\leq n$, we do not modify the sets~$\vec{X}$ and~$\vec{X}_0$. So, without loss of generality, we can assume that~$\vec{c}\in\setN^m$.

  We call~$P$ the set of pairs
  $(\vec{u},\vec{v})\in\setN^n\times\setN^m$ such that
  $\bigwedge_{i=1}^m\sum_{j=1}^na_{i,j}\vec{v}(j)=\vec{u}(i)$. We
  denote by~$P_0$ the set of minimal (for~$\leq$) non-zero pairs
  in~$P$.  In~\cite[Theorem~1]{pottier91} it is shown that
  every pair in~$P$ is a finite sum of pairs in~$P_0$, and moreover,
  every pair $(\vec{u},\vec{v})$ in~$P_0$ satisfies
  $\norm{\vec{u}}{}+\norm{\vec{v}}{}\leq C$ where
  $$C\eqdef (1+\max_{1\leq i\leq m}(\sum_{j=1}^n|a_{i,j}|+1))^m\;.$$
  
  Let $\vec{x}'\in \vec{X}$. Since $(\vec{c},\vec{x}')$ is in~$P$, it
  can be decomposed as a finite sum of pairs~in $P_0$. On the one
  hand, notice that the pairs $(\vec{u},\vec{v})$ with
  $\vec{u}=\vec{0}$ provides us with vectors $\vec{v}\in\vec{X}_0$
  satisfying $\norm{\vec{v}}{}\leq C$. On
  the other hand, notice that the decomposition of $(\vec{v},\vec{x})$
  cannot contains more that $\norm{\vec{c}}{}$ pairs
  $(\vec{u},\vec{v})$ with $\vec{u}\not=\vec{0}$. Notice that such a
  pair satisfies $\norm{\vec{v}}{}\leq C-1$.
  It follows that
  those pairs sum up to a pair $(\vec{c},\vec{x})$ in~$P$ such
  that $\norm{\vec{x}}{}\leq \norm{\vec{c}}{}\cdot (C-1)$. As
  $(\vec{c},\vec{x})\in P$, it follows that $\vec{x}\in\vec{X}$.  This
  concludes the proof.
\end{proof}

\begin{corollary}\label{lem:strongunbounded}
  Every model $\vec{h}$ of $E_\xi$ can be
  decomposed as the sum of a model $\vec{h}'$ of $E_\xi$ and a finite sum of
  models $\vec{h}_0$ of $E_\xi^0$ such that
  \begin{align*}
    \norm{\vec{h}'}{}&\leq \size{\xi}^{\size{\xi}-1}\;,&
    \norm{\vec{h}_0}{}&\leq \size{\xi}^{\size{\xi}-3}\;.
  \end{align*}
\end{corollary}
\begin{proof}
  Just apply \cref{lem:pottier} where $(a_{i,j})_{i,j}$ corresponds to
  the coefficients occurring in the characteristic system of $\xi$ in
  front of variables, and~$\vec{c}$ corresponds to the constant terms.
  Observe that $2+\max_{1\leq i\leq m}\sum_{j=1}^n|a_{i,j}|\leq \size{\xi}$,
  $m\leq \size{\xi}-3$, and $\norm{\vec{c}}{}\leq \size{\xi}^2$.
\end{proof}

\unbounded*
\begin{proof}
  This is a direct consequence of \cref{lem:strongunbounded} by
  observing that if~$\vec{h}$ is a model of~$E_\xi$ and~$\vec{h}_0$ is
  a model of~$E_\xi^0$ then $\vec{h}+n\vec{h}_0$ is a model of~$E_\xi$
  for every~$n\in\setN$.
\end{proof}

\section{Normal KLM Sequences}
\label{sec-normal}
In this section, we prove \cref{lem:normalnonempty}.  Throughout this
appendix, we assume that~$\xi$ denotes a \emph{normal} KLM sequence of
the form
$(\vec{x}_0G_0\vec{x}_1)\vec{a}_1\ldots (\vec{x}_k,G_k,\vec{y}_k)$,
where $G_j=(Q_j,q_{\mathit{in},j},q_{\mathit{out},j},T_j)$ for each $0\leq j\leq k$.

\subsection{Models of Normal KLM Sequences}

\begin{claim}\label{lem:h}
  There exists a model $\vec{h}$ of $E_\xi$ such that $\phi_j^{\vec{h}}(t)>0$ for
  every $t\in T_j$, and such that $\norm{\vec{h}}{}\leq 2\size{\xi}^{\size{\xi}-1}$.
\end{claim}
\begin{proof}
  As $\xi$ is satisfiable, there exists a model $\vec{h}$ of $E_\xi$. By
  decomposing $\vec{h}$ thanks to \cref{lem:strongunbounded}, we can
  assume that $\norm{\vec{h}}{}\leq\size{\xi}^{\size{\xi}-1}$. Moreover,
  since $\xi$ is unbounded, \cref{lem:strongunbounded} shows that
  for every $0\leq j\leq k$ and for every $t\in T_j$, 
  there exists a model $\vec{h}_0$ of $E_\xi^0$
  such that $\norm{\vec{h}_0}{}\leq
  \size{\xi}^{\size{\xi}-3}$ and such that $\phi_j^{\vec{h}_0}(t)>0$. By
  adding to $\vec{h}$, at most $\sum_{j=0}^k|T_j|$ models of $E_\xi^0$, we
  get a model of $E_\xi$ satisfying the claim.
\end{proof}

\begin{claim}\label{lem:h0}
  There
  exists a model $\vec{h}_0$ of $E_\xi^0$ such that for every $0\leq j\leq
  k$, for every $1\leq i\leq d$, and for every $t\in T_j$,
  \begin{itemize}
  \item if $\vec{x}_j(i)=\omega$ then $\vec{m}_j^{\vec{h}_0}(i)>0$,
  \item if $\vec{y}_j(i)=\omega$ then $\vec{n}_j^{\vec{h}_0}(i)>0$, and
  \item $\phi_j^{\vec{h}_0}(t)>0$,
  \end{itemize}
  and moreover,
  $$\norm{\vec{h}_0}{}\leq \size{\xi}^{\size{\xi}-2}\;.$$
\end{claim}
\begin{proof}
  Since $\xi$ is saturated, \cref{lem:strongunbounded} shows that
  for every
  $i\in\{1,\ldots,d\}$ and $j\in \{0,\ldots,k\}$:
  \begin{itemize}
  \item if
    $\vec{x}_j(i)=\omega$ then there exists a model $\vec{h}_0$ of $E_\xi^0$
    such that $\norm{\vec{h}_0}{}\leq
    \size{\xi}^{\size{\xi}-3}$ and such that $\vec{m}_j^{\vec{h}_0}(i)>0$, and
  \item if $\vec{y}_j(i)=\omega$ then there exists a model $\vec{h}_0$ of $E_\xi^0$
    such that $\norm{\vec{h}_0}{}\leq
    \size{\xi}^{\size{\xi}-3}$ and such that $\vec{n}_j^{\vec{h}_0}(i)>0$.
  \end{itemize}
  Moreover, since $\xi$ is unbounded, \cref{lem:strongunbounded}
  shows that for every $j\in \{0,\ldots,k\}$
  and for every $t\in T_j$ there exists a model $\vec{h}_0$ of $E_\xi^0$
  such that $\norm{\vec{h}_0}{}\leq
  \size{\xi}^{\size{\xi}-3}$ and such that $\phi_j^{\vec{h}_0}(t)>0$.
  
  \medskip

  It follows that by summing up at most $2d(k+1)+\sum_{0\leq j\leq
    k}|T_j|$ models of $E_\xi^0$, we get a model $\vec{h}_0$ of $E_\xi^0$
  satisfying the claim.
\end{proof}

\subsection{Flow Functions}\label{sub:flow}

\begin{claim}\label{lem:satpumred}
  For all $0\leq j\leq k$, there exists a function
  $F_j{:}\,Q_j\rightarrow \setN_\omega^d$ such that
  $F_j(q)=F_j(p)+\vec{a}$ for every transition $(p,\vec{a},q)\in T$,
  and such that $F_j(q_{\mathit{in},j})=\Facc{G_j}{\vec{x}_j}$ and
  $F_j(q_{\mathit{out},j})=\Bacc{G_j}{\vec{y}_j}$.
\end{claim}
\begin{proof}
  It suffices to prove the claim for some KLM triple $\vec{x}G\vec{y}$
  which is pumpable, rigid, and saturated.  Let us prove
  that for every $i\in\{1,\ldots,d\}$, there exists a function
  $f_i{:}\,Q\rightarrow\setN_\omega$ such that
  $f_i(q)=f_i(p)+\vec{a}(i)$ for every transition $(p,\vec{a},q)\in T$
  and such that $f_i(q_\mathit{in})=\Facc{G}{\vec{x}}(i)$ and
  $f_i(q_\mathit{out})=\Bacc{G}{\vec{y}}(i)$. Let
  $i\in\{1,\ldots,d\}$. Observe that if~$i$ is not fixed by $G$, then
  $\Facc{G}{\vec{x}}(i)=\omega=\Bacc{G}{\vec{y}}(i)$, and we can
  let~$f_i$ be the constant function mapping to~$\omega$.  Otherwise,
  if~$i$ is not fixed, then $\Facc{G}{\vec{x}}(i)=\vec{x}(i)$ and
  $\Bacc{G}{\vec{y}}(i)=\vec{y}(i)$. Since $\vec{x}G\vec{y}$ is
  saturated, we deduce that
  $\vec{x}(i)\in\setN\Longleftrightarrow\vec{y}(i)\in\setN$. Note
  that, if $\vec{x}(i)=\omega=\vec{y}(i)$, then we can let~$f_i$ be
  the constant function mapping to~$\omega$. So, we can assume that
  $\vec{x}(i),\vec{y}(i)\in\setN$.  Since $\vec{x}G\vec{y}$ is rigid,
  we deduce that there exists a function~$g_i{:}\,Q\rightarrow\setN$
  such that $g_i(q)=g_i(p)+\vec{a}(i)$ for every transition
  $(p,\vec{a},q)\in T$, and such that
  $g_i(q_\mathit{in})\sqsubseteq \vec{x}(i)$ and
  $g_i(q_\mathit{out})\sqsubseteq\vec{y}(i)$.  As
  $\vec{x}(i),\vec{y}(i)\in\setN$, we deduce that
  $g_i(q_\mathit{in})=\vec{x}(i)$ and
  $g_i(q_\mathit{out})=\vec{y}(i)$.  It follows that we can let
  $f_i\eqdef g_i$ in that case.
\end{proof}

\subsection{Pumping in Normal KLM Sequences}
Let us introduce the \emph{acceleration operator} $\nabla$ that maps a
pair of configurations $(\vec{x},\vec{x}')$ such that $\vec{x}\leq
\vec{x}'$ into the configuration $\vec{x}\nabla\vec{x}'$ defined for
every $1\leq i\leq d$ by:
$$(\vec{x}\nabla\vec{x}')(i)\eqdef\begin{cases}\omega & \text{ if
  }\vec{x}(i)<\vec{x}'(i)\\
  \vec{x}(i) & \text{ otherwise}
\end{cases}$$

\begin{claim}\label{lem:uvxy}
  There exists a sequence $(u_j,v_j)_{0\leq j\leq k}$ of pairs of
  words such that $|u_j|,|v_j|<\size{\xi}^{(d+1)^{d+1}}$, and a
  sequence $(\vec{x}_j',\vec{y}_j')_{0\leq j\leq k}$ of pairs of
  configurations $\vec{x}'_j\geq
  \vec{x}_j$ and $\vec{y}'_j\geq \vec{y}_j$ such that for every $0\leq
  j\leq k$:
  \begin{itemize}
  \item
    $q_{\mathit{in},j}(\vec{x}_j)\step[G_j]{u_j}q_{\mathit{in},j}(\vec{x}'_j)$ and
    $\Facc{G_j}{\vec{x}_j}=\vec{x}_j\nabla\vec{x}'_j$,
  \item $q_{\mathit{out},j}(\vec{y}'_j)\step[G_j]{v_j}q_{\mathit{out},j}(\vec{y}_j)$ and
    $\Bacc{G_j}{\vec{y}_j}=\vec{y}_j\nabla\vec{y}'_j$.
  \end{itemize}
\end{claim}
\begin{proof}
  Consider some triple $\vec x G\vec y\eqdef\vec x_jG_j\vec y_j$ for
  some $0\leq j\leq k$; since $\xi$ is normal, this triple is
  pumpable.  We just provide a proof for $u$ and $\vec{x}'$ since $v$
  and $\vec{y}'$ can be obtained by symmetry.  Let $I$ be the set of
  components $i\in\{1,\dots,d\}$ such that $\vec{x}(i)\in\setN$ and
  $\Facc{G}{\vec{x}}(i)=\omega$.  Fix some $i\in I$.  Notice that
  there exists a cycle $\theta_i$ on $q_\mathit{in}$ labelled by a
  word $\sigma_i$, and a configuration $\vec{x}_i\geq\vec{x}$ such
  that $\vec{x}\step{\sigma_i}\vec{x}_i$ and
  $\vec{x}_i(i)>\vec{x}(i)$.
  
  \medskip

  Let us prove that for every $n\in\setN$, there exists a
  configuration $\vec{c}\geq (n,\ldots,n)$ such that
  $q_\mathit{in}(\vec{x}|_I)\step[G]{*}q_\mathit{in}(\vec{c})$.
  Notice that
  $q_\mathit{in}(\vec{x}|_I)\step[G]{\sigma_i}q_\mathit{in}(\vec{x}_i|_I)$. From
  $\vec{x}|_I\leq \vec{x}_i|_I$, we deduce that there exists a
  configuration $\vec{c}_i\in\setN^d$ such that
  $\vec{x}_i|_I=\vec{x}|_I+\vec{c}_i$. As
  $\vec{x}|_I(i)<\vec{x}_i|_I(i)$, we get $\vec{c}_i(i)>0$. By
  monotony, notice that we have
  $q_\mathit{in}(\vec{x}|_I)\step[G]{*}q_\mathit{in}(\vec{c})$ where $\vec{c}\eqdef\vec{x}|_I+\sum_{i\in
    I}n\vec{c}_i$. Notice that $\vec{c}\geq (n,\ldots,n)$.

  \medskip

  By selecting $n$ large enough, and letting $C\eqdef\size{\xi}$,
  \cref{lem:rack} shows that 
  there exists another
  configuration~$\vec{c}$ and a word~$u$ such that
  $q_\mathit{in}(\vec{x}|_I)\step[G]{u}q_\mathit{in}(\vec{c})$ and such that
  $\vec{c}(i)\geq C-\size{G}$ for every $i$, and
  such that $|u|< C^{(d+1)^{d+1}}$.

  \medskip

  Let us prove that there exists a configuration $\vec{x}'$ such that
  $q_\mathit{in}(\vec{x})\step[G]{u}q_\mathit{in}(\vec{x}')$. Let $v$ be a
  prefix of $u$ and let us prove that $\vec{x}(i)+\Delta(u)(i)\geq 0$
  for every~$i$. If~$i\in I$, the previous paragraph provides the
  bound. If $\vec{x}(i)=\omega$, the proof is immediate. If $i\not\in
  I$ and $\vec{x}(i)\not=\omega$, then the function~$F_j$ introduced
  in \cref{lem:satpumred} shows that
  $\vec{x}(i)+\Delta(u)(i)=F_j(q)(i)\geq 0$ where $q$ is the state
  reached after reading $v$ from $q_{\mathit{in},j}$. Hence, we have proven the existence
  of $\vec{x}'$. Notice that $\vec{x}'(i)=\vec{x}(i)$ if $i\not\in I$
  and $\vec{x}'(i) \geq C-\size{G}>\vec{x}(i)$ for every $i\in
  I$.  We have proven the claim.
\end{proof}

\subsection{Proof of \Cref{lem:normalnonempty}}

\normalnonempty*
\begin{proof}
  We use the models $\vec{h}$ and $\vec{h}_0$ of~$E_\xi$ and~$E^0_\xi$
  defined in \cref{lem:h,lem:h0}, and the sequence
  $(u_j,v_j)_{0\leq j\leq k}$ and
  $(\vec x'_j,\vec y'_j)_{0\leq j\leq k}$ defined in \cref{lem:uvxy}.

  Now, let $\psi_{u_j}$ be the Parikh image of a cycle in~$G_j$
  on~$q_{\mathit{in},j}$ labelled by~$u_j$, and let~$\psi_{v_j}$ be the Parikh
  image of a cycle in~$G_j$ on~$q_{\mathit{out},j}$ labelled by~$v_j$.  We
  define $\phi_j\eqdef r\phi_j^{\vec{h}_0}-(\psi_{u_j}+\psi_{v_j})$
  where $r\eqdef 2\size{\xi}^{1+(d+1)^{d+1}}$.  Observe that
  $\phi_j(t)>0$ for every~$t\in T_j$.  Moreover, as~$\phi_j$ satisfies
  the homogeneous Kirchhoff system~$K_{G_j}^0$ and~$G_j$ is strongly
  connected, Euler's Lemma shows that there exists a cycle
  on~$q_{\mathit{in},j}$ labelled by some word~$w_j$ with a Parikh image equals
  to~$\phi_j$. Notice that
  $|w_j|=\sum_{t\in T_j}\phi_j(t)\leq r\size{\xi}^{\size{\xi}-3}$. Let
  $s\eqdef r\size{\xi}^{\size{\xi}-2}$, thus such that
  $\size{\xi}|w_j|\leq s$.  In particular
  $\Delta(w_j)=\Delta(\phi_j)=r\Delta(\phi_j^{\vec{h}_0})-(\Delta(u_j)+\Delta(v_j))$. From
  $\Delta(\phi_j^{\vec{h}_0})=\vec{n}_j^{\vec{h}_0}-\vec{m}_j^{\vec{h}_0}$
  we deduce
  \begin{equation}\vec{n}_j^{0}=\vec{m}_j^0+\Delta(w_j)\end{equation}
where $\vec{m}_j^0\eqdef r\vec{m}_j^{\vec{h}_0}+\Delta(u_j)$ and $\vec{n}_j^0\eqdef
r\vec{n}_j^{\vec{h}_0}-\Delta(v_j)$.

\medskip

Let $I_j$ be the set of components fixed by $G_j$.  Let us prove that
$\vec{m}_j^0,\vec{n}_j^0\in\setN^d$ and
$\vec{m}_j^0(i),\vec{n}_j^0(i)>0$ for every $i\not\in I_j$.  Observe
that if $i\in I_j$ then $\Delta(u_j)(i)=0$ since~$u_j$ is the label of
a cycle and in particular
$\vec{m}_j^0(i)=r\vec{m}_j^{\vec{h}_0}(i)\geq 0$.  If $i\not\in I_j$
and $\vec{x}_j(i)\in\setN$, because $\Facc{G_j}{\vec{x}_j}(i)=\omega$
we known that $\Delta(u_j)(i)>0$.  If $i\not\in I_j$ and
$\vec{x}_j(i)=\omega$, then $\vec{m}_j^{\vec{h}_0}(i)>0$ and in
particular
$\vec{m}_j^0(i)\geq r+\Delta(u_j)(i)\geq r+\Delta(u_j)(i)\geq 1$ by
definition of~$r$ and since $|u_j|<\size{\xi}^{(d+1)^{d+1}}$. We have
proven that $\vec{m}_j^0\geq \vec{0}$ and $\vec{m}_j^0(i)>0$ for every
$i\not\in I_j$.  Symmetrically, we see that $\vec{n}_j^0\in\setN^d$ and
$\vec{n}_j^0(i)>0$ for every $i\not\in I_j$.

\medskip

Notice that $\vec{x}_j=\vec{m}_j^{\vec{h}}+\omega\vec{m}_j^{\vec{h}_0}$. Since~$u_j$
is fireable from $\vec{x}_j$ and $\size{\xi}|u_j|\leq r$, we
deduce that~$u_j$ is fireable from
$\vec{m}_j^{\vec{h}}+r\vec{m}_j^{\vec{h}_0}$, thus
\begin{align}
  q_{\mathit{in},j}(\vec{m}_j^{\vec{h}}+r\vec{m}_j^{\vec{h}_0})&\step[G_j]{u_j}q_{\mathit{in},j}(\vec{m}_j^{\vec{h}}+\vec{m}_j^{0})\;.
\intertext{By monotony, this means that}
q_{\mathit{in},j}(\vec{m}_j^{\vec{h}}+sr\vec{m}_j^{\vec{h}_0})&\step[G_j]{u_j^s}q_{\mathit{in},j}(\vec{m}_j^{\vec{h}}+s\vec{m}_j^{0})\;,
\intertext{and by symmetry, we also have}
q_{\mathit{out},j}(\vec{n}_j^{\vec{h}}+s\vec{n}_j^{0})&\step[G_j]{v_j^s}q_{\mathit{out},j}(\vec{n}_j^{\vec{h}}+sr\vec{n}_j^{\vec{h}_0})\;.
\intertext{Moreover, since $\size{\xi}|w_j|\leq s$ and
$\vec{m}_j^0(i),\vec{n}_j^0(i)>0$ for every $i\not\in I_j$, we deduce}
q_{\mathit{in},j}(\vec{m}_j^{\vec{h}}+s\vec{m}_j^{0})&\step{w_j^s}q_{\mathit{in},j}(\vec{m}_j^{\vec{h}}+s\vec{n}_j^{0})\;.
\end{align}

Observe that $\phi_j^{\vec{h}}$ satisfies the Kirchhoff
system~$K_{G_j}$ and~$\phi_j^{\vec{h}}(t)>0$ for every~$t\in T_j$.
As~$G_j$ is strongly connected, Euler's Lemma shows
that~$\phi_j^{\vec{h}}$ is the Parikh image of a path from~$q_{\mathit{in},j}$
to~$q_{\mathit{out},j}$ labelled by a word~$\sigma_j$. The function $F_j$
introduced in \cref{lem:satpumred} shows that we have:
$$q_{\mathit{in},j}(\Facc{G_j}{\vec{x}_j})\xrightarrow[G_j]{\sigma_j}q_{\mathit{out},j}(\Bacc{G_j}{\vec{y}_j})$$
Notice that
$\Facc{G_j}{\vec{x}_j}=\vec{m}_j^{vec{h}}+\omega\vec{n}_j^0$ and
$\Bacc{G_j}{\vec{y}_j}=\vec{n}_j^{\vec{h}}+\omega\vec{n}_j^0$. Since $|\sigma_j|\leq \sum_{t\in T_j}\phi_j^{\vec{h}}(t)\leq
\norm{\vec{h}}{}\leq 2\size{\xi}^{2+\size{\xi}}$.
It follows that $\size{\xi}|\sigma_j|\leq s$. We deduce:
\begin{align}
q_{\mathit{in},j}(\vec{m}_j^{\vec{h}}+s\vec{n}_j^{0})&\step[G_j]{\sigma_j}q_{\mathit{out},j}(\vec{n}_j^{\vec{h}}+s\vec{n}_j^{0})\;.
\intertext{Thus, for every $0\leq j\leq k$,}
q_{\mathit{in},j}(\vec{m}_j^{\vec{h}}+sr\vec{m}_j^{\vec{h}_0})&\step[G_j]{u_j^sw_j^s\sigma_jv_j^s}q_{\mathit{out},j}(\vec{n}_j^{\vec{h}}+sr\vec{n}_j^{\vec{h}_0})\;.
\end{align}
This entails that $\sigma\eqdef(u_0^sw_0^s\sigma_0v_0^s)\vec{a}_1\ldots
(u_k^sw_k^s\sigma_kv_k^s)$ is in $L_\xi$. Notice that $|\sigma|\leq
k+(k+1)\cdot
s\cdot(2\size{\xi}^{(d+1)^{d+1}}+2\size{\xi}^{\size{\xi}-1}+r\size{\xi}^{\size{\xi}-3})\leq
7\size{\xi}^{2(d+1)^{d+1}+2\size{\xi}-1}$. Observe that
$2(d+1)^{d+1}\leq \size{\xi}$ and $7\leq \size{\xi}$. Hence
$|\sigma|\leq \size{\xi}^{3\size{\xi}}$.
\end{proof}

\end{document}